\documentclass[journal,twoside]{IEEEtranTCOM}
\normalsize

\usepackage{cite}
\ifCLASSINFOpdf
  \usepackage[pdftex]{graphicx}
  \usepackage{epstopdf}
  \graphicspath{{./}}
  \DeclareGraphicsExtensions{.eps}
\else
  \usepackage[dvips]{graphicx}
  \graphicspath{{./figs/}}
  \DeclareGraphicsExtensions{.eps}
\fi
\usepackage{mathtools}
\usepackage{amssymb, amsmath, amsthm, amsfonts}
\usepackage{algorithm}
\usepackage{algorithmic}
 \usepackage[T1]{fontenc}
 
\usepackage{array}
\ifCLASSOPTIONcompsoc
  \usepackage[caption=false,font=normalsize,labelfont=sf,textfont=sf]{subfig}
\else
  \usepackage[caption=false,font=footnotesize]{subfig}
\fi
\usepackage{dblfloatfix}
\usepackage{url}
\usepackage{bm}

\usepackage{xcolor}
\usepackage{soul}
\sethlcolor{yellow} 
\setul{0.3ex}{0.1ex}
\newcommand\bmcal[1]{\bm{\mathcal{#1}}}

\DeclareMathSymbol{\mlq}{\mathord}{operators}{``}
\DeclareMathSymbol{\mrq}{\mathord}{operators}{`'}

\DeclareMathOperator*{\argmin}{argmin}

\newcommand\eqnum{\addtocounter{equation}{1}\tag{\theequation}}

\newtheorem{thm}{Theorem}

\begin{document}

%
% paper title
% Titles are generally capitalized except for words such as a, an, and, as,
% at, but, by, for, in, nor, of, on, or, the, to and up, which are usually
% not capitalized unless they are the first or last word of the title.
% Linebreaks \\ can be used within to get better formatting as desired.
% Do not put math or special symbols in the title.
\title{Prefix-Free Code Distribution Matching for Probabilistic Constellation Shaping}
%
%
% author names and IEEE memberships
% note positions of commas and nonbreaking spaces ( ~ ) LaTeX will not break
% a structure at a ~ so this keeps an author's name from being broken across
% two lines.
% use \thanks{} to gain access to the first footnote area
% a separate \thanks must be used for each paragraph as LaTeX2e's \thanks
% was not built to handle multiple paragraphs
%

\author{Junho~Cho,~\IEEEmembership{Member,~IEEE}
%        John~Doe,~\IEEEmembership{Fellow,~OSA,}
%        and~Jane~Doe,~\IEEEmembership{Life~Fellow,~IEEE}% <-this % stops a space
%\thanks{M. Shell was with the Department of Electrical and Computer Engineering, Georgia Institute of Technology, Atlanta, GA, 30332 USA e-mail: (see http://www.michaelshell.org/contact.html).}% <-this % stops a space
%\thanks{J. Doe and J. Doe are with Anonymous University.}% <-this % stops a space
%\thanks{Manuscript received April 19, 2005; revised August 26, 2015.}}
%\thanks{J. Cho is with Nokia Bell Labs, Holmdel, NJ, 07733 USA e-mail: junho.cho@nokia-bell-labs.com}
}% <-this % stops a space

% note the % following the last \IEEEmembership and also \thanks - 
% these prevent an unwanted space from occurring between the last author name
% and the end of the author line. i.e., if you had this:
% 
% \author{....lastname \thanks{...} \thanks{...} }
%                     ^------------^------------^----Do not want these spaces!
%
% a space would be appended to the last name and could cause every name on that
% line to be shifted left slightly. This is one of those "LaTeX things". For
% instance, "\textbf{A} \textbf{B}" will typeset as "A B" not "AB". To get
% "AB" then you have to do: "\textbf{A}\textbf{B}"
% \thanks is no different in this regard, so shield the last } of each \thanks
% that ends a line with a % and do not let a space in before the next \thanks.
% Spaces after \IEEEmembership other than the last one are OK (and needed) as
% you are supposed to have spaces between the names. For what it is worth,
% this is a minor point as most people would not even notice if the said evil
% space somehow managed to creep in.

% The paper headers
\markboth{IEEE Transactions on Communications}%
{Submitted paper}
% The only time the second header will appear is for the odd numbered pages
% after the title page when using the twoside option.
% 
% *** Note that you probably will NOT want to include the author's ***
% *** name in the headers of peer review papers.                   ***
% You can use \ifCLASSOPTIONpeerreview for conditional compilation here if
% you desire.

% If you want to put a publisher's ID mark on the page you can do it like
% this:
%\IEEEpubid{0000--0000/00\$00.00~\copyright~2015 IEEE}
% Remember, if you use this you must call \IEEEpubidadjcol in the second
% column for its text to clear the IEEEpubid mark.

% use for special paper notices
%\IEEEspecialpapernotice{(Invited Paper)}

% make the title areaadaptabilityt
\maketitle

% As a general rule, do not put math, special symbols or citations
% in the abstract or keywords.

\vspace{-4em}

\begin{abstract}
In this work, we construct variable-length prefix-free codes that are optimal (or near-optimal) in the sense that no (or few) other codes of the same cardinality can achieve a smaller expected energy per code symbol for the same resolution rate. Under stringent constraints of 4096 codewords or below per codebook, the constructed codes yield an energy per code symbol within a few tenths of a dB of the \emph{un}constrained theoretic lower bound, across a wide range of resolution rates with fine granularity. We also propose a framing method that allows variable-length codes to be transmitted using a fixed-length frame. The penalty caused by framing is studied using simulations and analysis, showing that the energy per code symbol is kept within 0.2~dB of the unconstrained theoretic limit for some tested codes with a large frame length. When the proposed method is used to implement probabilistic constellation shaping for communications in the additive white Gaussian noise channel, simulations show that between 0.21~dB and 0.98~dB of shaping gains are achieved relative to uniform 4-, 8-, 16- and 32-quadrature amplitude modulation.
\end{abstract}

% Note that keywords are not normally used for peerreview papers.
%\begin{IEEEkeywords}
%Prefix-free codes, variable-length codes, distribution matching, probabilistic constellation shaping.
%\end{IEEEkeywords}

% For peer review papers, you can put extra information on the cover
% page as needed:
% \ifCLASSOPTIONpeerreview
% \begin{center} \bfseries EDICS Category: 3-BBND \end{center}
% \fi
%
% For peerreview papers, this IEEEtran command inserts a page break and
% creates the second title. It will be ignored for other modes.
\IEEEpeerreviewmaketitle

%\noindent1) Introduction\\
%2) Problem formulation\\
%3) Code construction\\
%3-1) F2V Code\\
%3-2) F2V Code\\
%3-3) F2V Code\\
%4) Framing\\
%4-1) Method\\
%4-2) Monte-Carlo Simulation\\
%4-3) Gaussian Approximation\\
%5) Conclusion

\section{Introduction}
\label{sec:intro}
% The very first letter is a 2 line initial drop letter followed
% by the rest of the first word in caps.
% 
% form to use if the first word consists of a single letter:
% \IEEEPARstart{A}{demo} file is ....
% 
% form to use if you need the single drop letter followed by
% normal text (unknown if ever used by the IEEE):
% \IEEEPARstart{A}{}demo file is ....
% 
% Some journals put the first two words in caps:
% \IEEEPARstart{T}{his demo} file is ....
% 
% Here we have the typical use of a "T" for an initial drop letter
% and "HIS" in caps to complete the first word.
%\IEEEPARstart{P}{refix}-free codes are generally used in \emph{source coding} to transform \red{information symbols} of an arbitrary probability distribution to \red{code symbols} of a uniform probability distribution to maximizing entropy per \red{code symbol}.
\IEEEPARstart{P}{refix}-free codes are commonly used in lossless data compression to reduce the number of code symbols to describe information symbols, such that the original information symbols can be completely reconstructed from the code symbols.
Optimal prefix-free coding that achieves the minimum number of code symbols transforms information symbols of a non-uniform probability distribution into code symbols of a uniform probability distribution.
In this work, prefix-free codes perform a reverse operation called \emph{distribution matching (DM)}; namely, it transforms information symbols of a uniform probability distribution to code symbols of a desired probability distribution.
In particular, we restrict information symbols to bits $b\in\mathcal{B}:=\{0,1\}$ drawn with equal probabilities $\mathbb{P}(b=0) = \mathbb{P}(b=1)= 0.5$, and construct prefix-free codes with a unipolar $M$-ary amplitude shift keying ($M$-ASK) code alphabet $\mathcal{X}_{M\textrm{-ASK}} = \{1,3,\ldots,2M-1\}$.
Defining the \emph{resolution rate} as the expected number of information bits per code symbol\cite{Bocherer13fixed}, the constructed prefix-free codes have the minimum or near-minimum average symbol energy among all possible prefix-free codes of the same cardinality for the same resolution rate.
We also propose a \emph{framing} method, with which variable-length prefix-free codewords can always be contained in a fixed-length frame, thereby facilitating application of prefix-free codes to communications. 

A prominent application of \emph{prefix-free code distribution matching (PCDM)} is \emph{probabilistic constellation shaping (PCS)} for capacity-approaching communications.
PCS makes low-energy symbols in the code alphabet appear with a higher probability than high-energy symbols, thereby reducing the average transmit energy for the same \emph{information rate (IR)} \cite{Tunstall67,Varn71,Hu71,Forney84 ,Calderbank90,Forney92,Livingston92 ,Kschischang93,Raphaeli04,Kaimalettu07,Valenti12}.
%%%%%%
%\begin{figure}[!t]
%\centering
%\captionsetup[subfloat]{captionskip=0pt}
%%\subfloat[][]{\!\includegraphics[width=0.92\linewidth]{fig_archi_A}\label{fig:archi_a}}\\[0pt]
%%\subfloat[][]{\!\includegraphics[width=0.92\linewidth]{fig_archi_B}\label{fig:archi_b}}\\[0pt]
%%\subfloat[][]{\!\includegraphics[width=0.92\linewidth]{fig_archi_C}\label{fig:archi_c}}\\
%\subfloat[][]{\!\includegraphics[width=0.46\linewidth]{fig_archi_A}\label{fig:archi_a}} \qquad
%\subfloat[][]{\!\includegraphics[width=0.46\linewidth]{fig_archi_B}\label{fig:archi_b}}\\[0pt]
%\subfloat[][]{\!\includegraphics[width=0.46\linewidth]{fig_archi_C}\label{fig:archi_c}}\\
%\caption{Architectures for a coded modulation system.}% 
%\label{fig:archi}
%\end{figure}
%%%%%%
A long-standing challenge of PCS has been the incorporation of forward error correction (FEC) coding and DM.
If DM is embedded between FEC encoding and decoding, errors occurred in the channel can be significantly boosted in the DM decoding process, thereby rendering the outer FEC coding unsuccessful.
For the case of variable-length DM, even a single error can cause insertion or deletion of information symbols, leading to synchronization errors or catastrophic propagation of errors.
In a reversed architecture where FEC coding is embedded between DM encoding and decoding, a symbol distribution formed by DM encoding is not preserved by the subsequent FEC encoding, since (linear) FEC encoding generates almost equiprobable code symbols from information bits of any distribution.
There have been approaches to jointly optimizing FEC and DM \cite{Forney92,Raphaeli04,Kaimalettu07,Valenti12}, but they lack rate adaptability and involve the design of customized FEC codes.
Recently, however, an architecture called \emph{probabilistic amplitude shaping (PAS)} \cite{Bocherer15} solved the problem.
In the PAS architecture, information bits are first encoded by DM to produce \emph{amplitudes} of transmit symbols with a desired probability distribution, then a following systematic FEC encoder generates parity bits that constitute \emph{signs} of the transmit symbols.
In this manner, the amplitude distribution remains unchanged by equiprobable parity bits, and the error propagation and synchronization errors do not occur since FEC decoding corrects errors before DM decoding.
PAS enables separate design of FEC and DM, hence allows for the use of off-the-shelf FEC codes in conjunction with independently optimized DM.
In this paper, the performance of the proposed PCDM in communication systems will be numerically evaluated using the PAS architecture in comparison with the conventional communication schemes whose transmit symbols are uniformly distributed.

%%%%%%%%%%%%%%%%%%%%%%%%%%%%%%%%%%%%%%%%%%%%%%%%%%%%%%%%%%%%%%%%%%%%%%%%%
%\section{Prefix-Free Codes}
%\label{sec:codes_intro}
\section{Prior Works}
\label{sec:prior_works}

\subsection{Fixed-Length DMs}

The $m$-out-of-$n$ coding~\cite{ramabadran1990coding} and the  constant composition DM (CCDM)~\cite{Schulte16} are fixed-to-fixed length (F2F) DMs that realize a target code symbol distribution on a block-by-block basis using arithmetic coding.
%Both schemes accommodate rate adaptation by creating an arbitrary symbol distribution. 
The CCDM was shown to be asymptotically optimal in the sense that the normalized Kullback-Leibler (KL) divergence between the desired and realized symbol distributions vanishes as the block length goes to infinity~\cite{Schulte16}.
The computational complexity of CCDM increases linearly with the block length, making it possible to use a very large block length to closely approach the asymptotic performance.
However, CCDM requires high-precision multiplications and divisions.
In~\cite{Fehenberger18multiset}, the multiset partitioning DM (MPDM) was proposed to enhance the degrading performance of CCDM in small blocks, by letting CCDM create different symbol distributions across multiple blocks that are averaged to the target distribution.
For a small block length, typically below 100 symbols, there have been algorithmic approaches to implement an F2F DM, such as the shell mapping~\cite{lang1989leech,laroia1994optimal,Schulte18shell} and the enumerative sphere shaping~\cite{Gultekin17}.
The algorithmic approaches substantially reduced the complexity of indexing a point in a sphere, by divide and conquer~\cite{lang1989leech,laroia1994optimal,Schulte18shell} or dynamic programming\cite{Gultekin17}, which is otherwise exponential in the dimension.
Nevertheless, their complexity is polynomial in the dimension, e.g., quadratic for the Algorithm 2 of~\cite{laroia1994optimal}, hence a large block length cannot be used.
More recently, an F2F DM was realized in~\cite{Yoshida18low} by using a distribution uniformizer that makes all the symbols in a code alphabet appear the same number of times in each block, followed by run-length coding that transforms the uniform distribution into a non-uniform target distribution.

%Note that the PCDM and CCDM have a linear complexity in the block length, making it possible to realize the asymptotic performance in a large block length, at the expense of a long latency.

%Both the shell mapping [Laroria, Lang and Longstaff, SMDM] and the enumerative amplitude shaping [Gultekin] are algorithmic approaches that significantly reduces the complexity of indexing and addressing a point in a sphere, which is otherwise exponential in the dimension.
%In [Gultekin], an enumerative amplitude shaping approach to realize an F2F DM was proposed, which significantly reduces the complexity of indexing and addressing a point in a sphere, which is otherwise exponential in the dimension.
%Nevertheless, the complexity still increases with the block length, hence this method cannot be used for large block lengths, e.g., greater than 100 symbols in a block.

\subsection{Variable-Length PCDM}

It was recently shown that F2F DMs like $m$-out-of-$n$~\cite{ramabadran1990coding} and CCDM~\cite{Schulte16} cannot fundamentally eliminate the \emph{un}-normalized KL divergence between the desired and realized symbol distributions~\cite{schulte2017divergence}.
On the other hand, \emph{variable-length} may achieve zero \emph{un}-normalized KL divergence~\cite{Bocherer13fixed}.
Importantly, variable-length PCDM has a linear complexity in the dimension, as in CCDM and the run-length coding~\cite{Yoshida18low}, hence can be realized in large dimensions to achieve zero un-normalized KL divergences.
Furthermore, unlike the $m$-out-of-$n$ or CCDM, PCDM does not require high-precision arithmetic.
To this end, we limit the scope to variable-length prefix-free codes and study the construction of prefix-free codes for rate-adaptable PCDM throughout the paper.

Let $\bm{b}$ and $\bm{x}$, respectively, denote a concatenation of an indefinite number of information symbols drawn from the binary alphabet $\mathcal{B}$ and a concatenation of an indefinite number of code symbols drawn from the ASK alphabet $\mathcal{X}$.
Then, a \emph{code} $\mathcal{C}: \bmcal{B} \mapsto \bmcal{X} $ defines a \emph{bijective} function that maps every information word $\bm{b}$ in the \emph{dictionary} $\bmcal{B}$ to a codeword $\bm{x}$ in the \emph{codebook} $\bmcal{X}$.
If the codeword $\bm{x} \in \bmcal{X}$ is an ordered concatenation of two words $\bm{x}_{\mathrm{pre}}$ and $\bm{x}'$, then $\bm{x}_{\mathrm{pre}}$ is called a \emph{prefix} of the word $\bm{x}$.
A code is a \emph{prefix-free code} if any codeword $\bm{x} \in \bmcal{X}$ is not a prefix of another codeword.
A prefix-free code is an \emph{instantaneously decodable code}, whose decoding can be performed as immediately as a codeword is found on successive receipt of the code symbols.

In a special case where all code symbols in $\mathcal{X}$ have equal transmit energies, and if the information word distribution is known a priori, \emph{Huffman codes} represent optimal prefix-free codes in the sense that no other codes can produce a shorter expected codeword length, hence a smaller average code symbol energy.
In a more general case where the energies of symbols in $\mathcal{X}$ are not necessarily equal, the prefix-free code construction problem has been addressed in the context of PCS~\cite{Abrahams98}, using \emph{Lempel-Even-Cohn (LEC)} coding~\cite{Lempel73} and \emph{Varn} coding~\cite{Varn71}.
LEC coding solves the problem when all codewords have an equal length.
LEC coding is isomorphic to Huffman coding if we convert the energy of a codeword into the probability of the corresponding information word.
A Huffman codebook therefore becomes the LEC dictionary for parsing information bits, mapped to an LEC codebook that consists of equal-length codewords, completing a \emph{variable-to-fixed length (V2F)} code.
However, while we want to minimize the energy per \emph{code} symbol for arbitrary resolution rates, LEC coding minimizes the energy per \emph{information} symbol, and its resolution rate is determined by the Huffman codebook that is given \emph{a priori}.
Varn coding provides another solution when information words have an equal length, hence constructing \emph{fixed-to-variable length (F2V)} codes.
However, Varn coding requires a non-uniform probability mass function (PMF) for the equal-length information words, which must be given \emph{a priori} to realize a particular resolution rate.
Another approach to constructing F2V codes is by minimizing the KL divergence between the desired and realized PMFs for a binary code alphabet \cite{Amjad13}, which also assumes \emph{a priori} knowledge of the optimal PMF for the target resolution rate.
In general cases where both information words and codewords can have unequal lengths, there are known algorithms to construct optimal prefix-free codes \cite{Hu71,Choi96,Golin96,Golin98,Golin08} using, e.g., dynamic programming \cite{Golin98,Golin08}.
However, they require unequal probabilities of the information symbols that must be given \emph{a priori}, and the resulting codes are optimal only in the sense that they minimize the average \emph{codeword} energy, without taking into account the resolution rate.
To the best of our knowledge, none of the existing approaches addressed the problem of constructing prefix-free codes to minimize the average code symbol energy for arbitrary target resolution rates.

%%%%%%%%%%%%%%%%%%%%%%%%%%%%%%%%%%%%%%%%%%%%%%%%%%%%%%%%%%%%%%%%%%%%%%%%%
\section{Prefix-Free Codes for Various Resolution Rates}
\label{sec:codes}

%%%
\begin{table}[!t]
\caption{Examples of prefix-free codes for 2-ASK code alphabet}
\label{tab:examples}
\centering
\vspace*{-1em}
\includegraphics[width=3.35in]{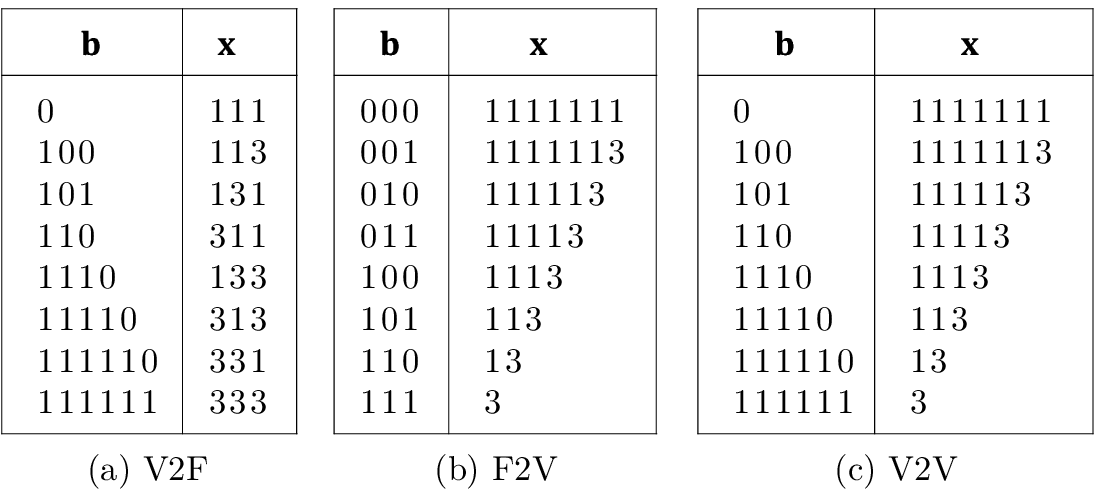}
\end{table}
%%%
\begin{figure}[!t]
\centering
\includegraphics[width=3.3in]{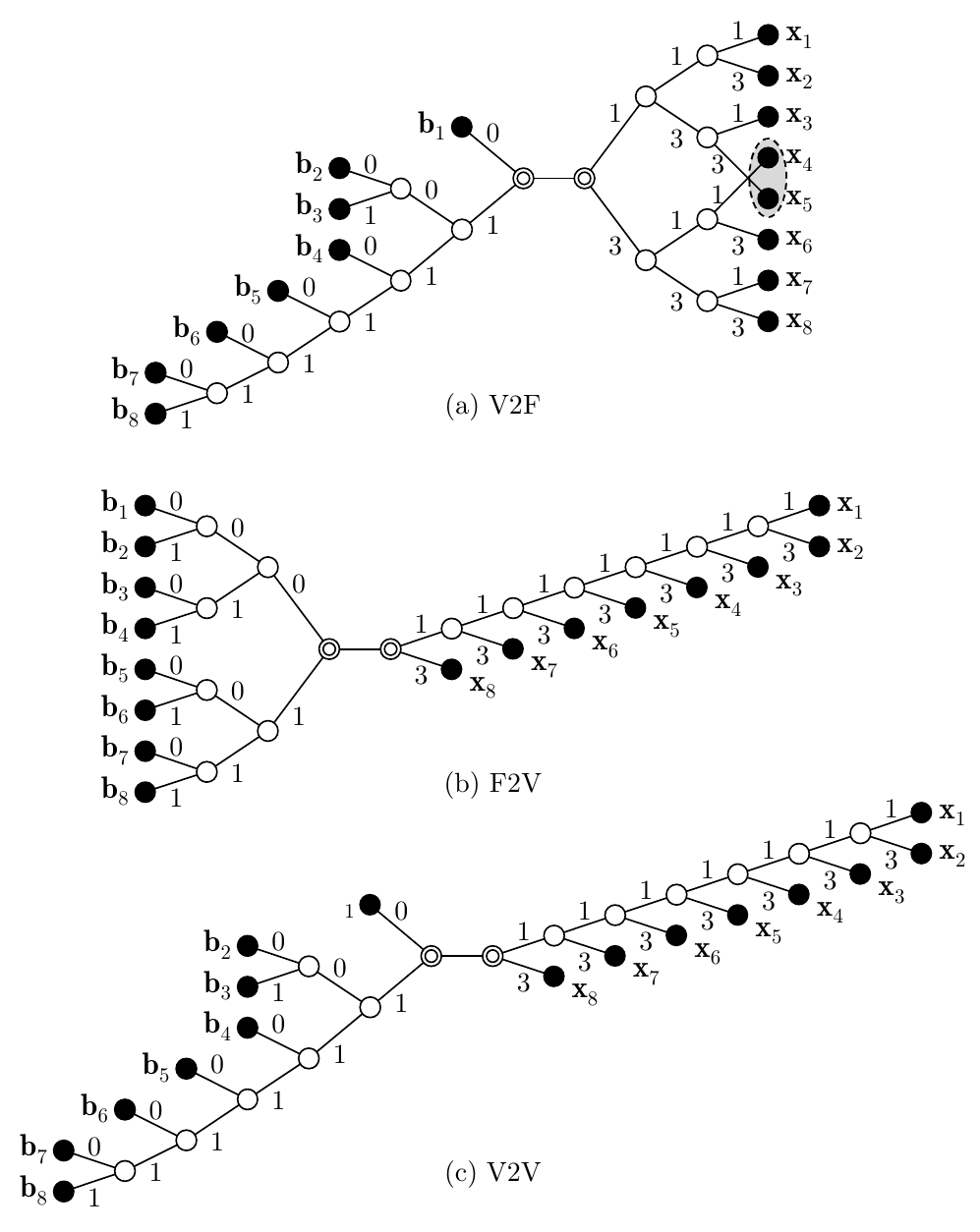}
\caption{Root-concatenated trees representing the codes in Tab.~\ref{tab:examples}.}
\label{fig:trees}
\end{figure}
%%%
Let $\bm{b}_n$ and $\bm{x}_n$ denote the $n$-th information word in $\bmcal{B}$ and the $n$-th codeword in $\bmcal{X}$, respectively, and let $l(\bm{a})$ denote the length of string $\bm{a}$.
Since mapping~$\mathcal{C}$ is bijective, we abuse notation to write $|\mathcal{C}|=|\bmcal{B}|=|\bmcal{X}|$.
Table~\ref{tab:examples} shows examples of V2F, F2V, and V2V codes with $|\mathcal{C}|=8$, for the 2-ASK code alphabet $\mathcal{X}_{2\textrm{-ASK}} = \{1,3\}$.
Let $\bm{l}(\bmcal{B}) := [l(\bm{b}_1), \ldots, l(\bm{b}_{N})]$ with $N = |\mathcal{C}|$ denote an information word length vector, $\bm{l}(\bmcal{X}) := [l(\bm{x}_1), \ldots, l(\bm{x}_{N})]$ a codeword length vector, and $\bm{e}(\bmcal{X}) := [||\bm{x}_1||^2, \ldots, ||\bm{x}_{N}||^2]$ a  codeword energy vector.
For example, the V2V code in Table~\ref{tab:examples}(c) has $\bm{l}(\bmcal{B}) = [1, 3, 3, 3, 4, 5, 6, 6]$, $\bm{l}(\bmcal{X}) = [7, 7, 6, 5, 4, 3, 2, 1]$, and $\bm{e}(\bmcal{X}) = [7, 15, 14, 13, 12, 11, 10, 9]$.
The tree diagrams for the codes in Table~\ref{tab:examples} are depicted in Fig.~\ref{fig:trees}, where double circles, open circles, and closed circles represent the root nodes, branch nodes, and leaf nodes, respectively, which will be defined below.

A code is formed by concatenating the roots of two \emph{ordered} trees, which we call \emph{left} and \emph{right trees} depending on their relative position.
We use the following terminologies throughout the paper to describe the structure of a right tree (defined in a similar fashion for a left tree):
\begin{itemize}
  \item \emph{Root}: The left-most node of a tree.
  \item \emph{Child}: A node directly connected to the right of a node.
  \item \emph{Parent}: The converse of a child.
  \item \emph{Siblings}: A group of nodes with the same parent.
  \item \emph{Branch}: A node with at least one child.
  \item \emph{Leaf}: A node with no children.
  \item \emph{Degree}: The number of children of a node.
  \item \emph{Path}: A sequence of nodes and edges connecting a node with another node.
  \item \emph{Depth}: The depth of a node is the number of edges from the root node to the node, i.e., the path length connecting the root node and the node.
  \item \emph{Height}: The height of a tree is the longest path length between the root and leaves.  
  \item \emph{Size}: The size of a tree is the number of all leaf nodes of the tree.
\end{itemize}
The $n$-th leaf of a left tree represents the information word $\bm{b}_n$ and that of a right tree represents the codeword $\bm{x}_n$.
The depth of the $n$-th leaf equals $l(\bm{b}_n)$ (in the left tree) or $l(\bm{x}_n)$ (in the right tree), and the probability mass vector of the leaves is given by $\bm{p} = [p_1,\ldots,p_N]^T$, with $p_n = \mathbb{P}_{\bm{B}}(\bm{b}_n) = \mathbb{P}_{\bm{X}}(\bm{x}_n)$.

Prefix-free code encoding is a random process that parses a random variable $\bm{B}$ which takes values $\bm{b}\in{\bmcal{B}}$ to produce a codeword $\bm{X} = \mathcal{C}(\bm{B}) \in \bmcal{X}$.
Since our information bits are assumed to be independent and equiprobable, the probability of information words is \emph{dyadic}, i.e., $\mathbb{P}_{\bm{B}}(\bm{b}) = 2^{-l(\bm{b})}$ with $\sum_{\bm{b} \in \bmcal{B}} \mathbb{P}_{\bm{B}}(\bm{b}) = 1$.
Therefore, the expected symbol energy can be calculated as
%%%
\begin{align*} 
\mathsf{E}(\bmcal{B,X}) &:= \frac{\mathbb{E}(||\bm{X}||^2)}{\mathbb{E}(l(\bm{X}))} %\\
		    = \frac{\sum_{\bm{x}\in\bmcal{X}} 2^{-l(\mathcal{C}^{-1}(\bm{x}))} ||\bm{x}||^2}{\sum_{\bm{x}\in\bmcal{X}} 2^{-l(\mathcal{C}^{-1}(\bm{x}))} l(\bm{x})},  			\eqnum \label{eqn:E}
%		    = \frac{\sum_{n=1}^{N} 2^{-u_n} w_n}{\sum_{n=1}^{N} 2^{-u_n} v_n},  			\eqnum \label{eqn:E}
\end{align*}
%%%
where $\mathbb{E}(\cdot)$ denotes expectation.
The resolution rate of this code can be calculated as
%%%
\begin{align*} 
\mathsf{R}(\bmcal{B,X}) &:= \frac{\mathbb{E}(l(\bm{B}))}{\mathbb{E}(l(\bm{X}))} %\\
		    = \frac{\sum_{\bm{b}\in\bmcal{B}} 2^{-l(\bm{b})} l(\bm{b})}{\sum_{\bm{x}\in\bmcal{X}} 2^{-l(\mathcal{C}^{-1}(\bm{x}))} l(\bm{x})}. \eqnum \label{eqn:R}
\end{align*}
%%%
Since it is apparent that $\mathsf{E}(\bmcal{B,X})$ and $\mathsf{R}(\bmcal{B,X})$ are functions of $\bmcal{B}$ and $\bmcal{X}$, we henceforth omit the arguments and denote them by $\mathsf{E}_\mathcal{C}$ and $\mathsf{R}_\mathcal{C}$ for simplicity.
The bijection $\mathcal{C}$ is immediately defined from the ordered sets $\bmcal{B}$ and $\bmcal{X}$, hence the optimal prefix-free coding problem for a target resolution rate $\mathsf{R}^*$ can be written as
%%%
\begin{align*}
\begin{array}{ll}
\underset{\bmcal{B},\bmcal{X}}{\text{minimize}} & \mathsf{E}_\mathcal{C} 		\eqnum \label{eqn:optimize_E} \\		
\text{subject to} & \mathsf{R}_\mathcal{C} \geq \mathsf{R}^*, \\ [-0.2ex]
			 	& \bmcal{B}\text{ and }\bmcal{X} \text{ are left and right trees, respectively.}
\end{array}
\end{align*}

The resolution rate is upper-bounded as $\mathsf{R}_\mathcal{C} \leq \mathbb{H}(X)$, where $X$ denotes an independent and identically distributed (IID) random variable that takes values in $\mathcal{X}$ according to the same distribution as the prefix-free code symbols, and where $\mathbb{H}(X)$ denotes the entropy of $X$.
Under an average energy constraint, the entropy is maximized by the Maxwell-Boltzmann (MB) distribution\cite{Kschischang93} $\mathbb{P}_{X}(x) := \exp(-\lambda |x|^2) / \sum_{x\in\mathcal{X}} \exp(-\lambda |x|^2), \lambda \geq 0$.
Conversely, the average symbol energy $\mathbb{E}(|X|^2)$ is minimized by the MB distribution to achieve a target entropy $\mathbb{H}(X)$.
Therefore, if we denote by $X_{\text{MB}}$ an IID random variable drawn according to the MB distribution that fulfills $\mathbb{H}(X) = \mathsf{R}_\mathcal{C}$, we obtain a lower bound of the average symbol energy to achieve the resolution rate $\mathsf{R}_\mathcal{C}$ as $\mathsf{E}_\mathcal{C} \geq \mathbb{E}(|X_{\text{MB}}|^2)$.
The energy efficiency of a prefix-free code that achieves $\mathsf{R}_\mathcal{C}$ can therefore be evaluated by the \emph{energy gap} defined as
\begin{align*}
\mathsf{E}_{\text{Gap}} := \mathsf{E}_\mathcal{C} / \mathbb{E}(|X_{\text{MB}}|^2). \label{eqn:E_gap}
\end{align*}

%%%%%%%%%%%%%%%%%%%%%%%%%%%%%%%%%%%%%%%%%%%%%%%%%%%%%%%%%%%%%%%%%%%%%%%%%
\subsection{V2F Codes}

To solve problem \eqref{eqn:optimize_E}, we begin with a \emph{balanced} $M$-ary right tree, i.e., a right tree in which every branch has $M$ children and every leaf is at the same depth (see the right tree of Fig.~\ref{fig:trees}~(a)), which represents a codebook $\bmcal{X}$ with a fixed codeword length $l(\bm{x}) = v$ for all $\bm{x}\in\bmcal{X}$.
Then, the codebook can immediately be obtained from the fixed codeword length $v$ by lexicographical ordering as $\bmcal{X}_v = \mathcal{X}_{M\textrm{-ASK}}^{v}$, and $|\mathcal{C}| = M^v$.
In this case, \eqref{eqn:E} and \eqref{eqn:R} degenerate, respectively, to 
\begin{align}
\mathsf{E}_\mathcal{C} = \mathbb{E}(||\bm{X}||^2)/v
\end{align}
and
\begin{align}
\mathsf{R}_\mathcal{C} = \mathbb{E}(l(\bm{B})) / v = \mathbb{H}(\bm{X})/v,
\end{align}
where the last equation holds since $\mathbb{P}_{\bm{X}}$ is dyadic.
Therefore, for V2F codes, minimizing the average code symbol energy is equivalent to minimizing the average codeword energy.
Also, maximizing $\mathsf{R}_\mathcal{C}$ subject to $\mathsf{E}_\mathcal{C} \leq \mathsf{E}^*$ is equivalent to maximizing $\mathbb{H}(\bm{X})$ subject to $\mathbb{E}(||\bm{X}||^2) \leq v\mathsf{E}^*$.

%%%%%%
\begin{figure}[!t]
\centering
\subfloat[][]{\!\includegraphics[width=0.47\linewidth]{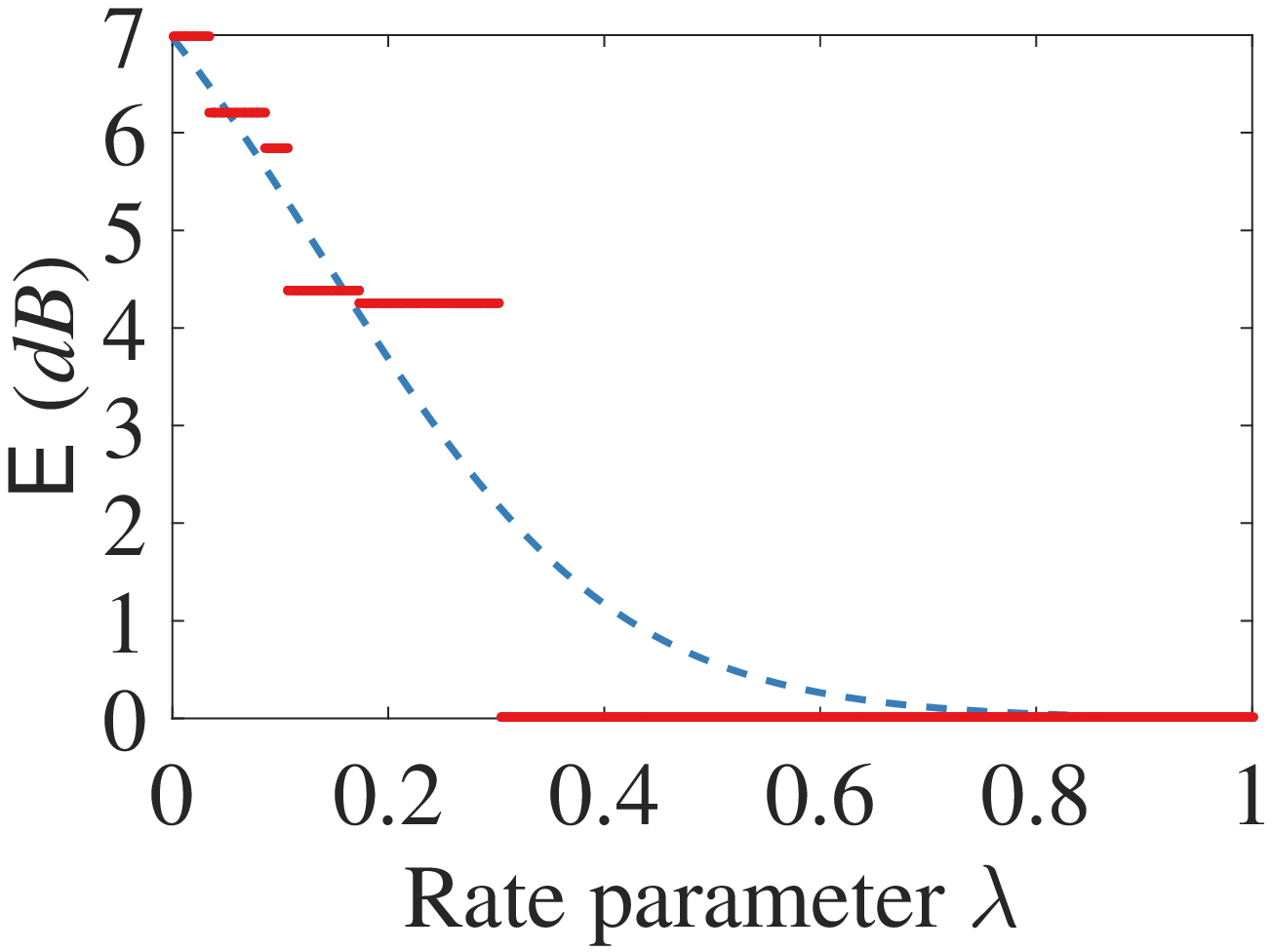}\label{fig:v2f_ghc_a}}%
%\subfloat[][]{\!\includegraphics[height=8.5em]{fig_lambda_E_cr.pdf}\label{fig:v2f_ghc_a}}%
\quad
\subfloat[][]{\!\includegraphics[width=0.47\linewidth]{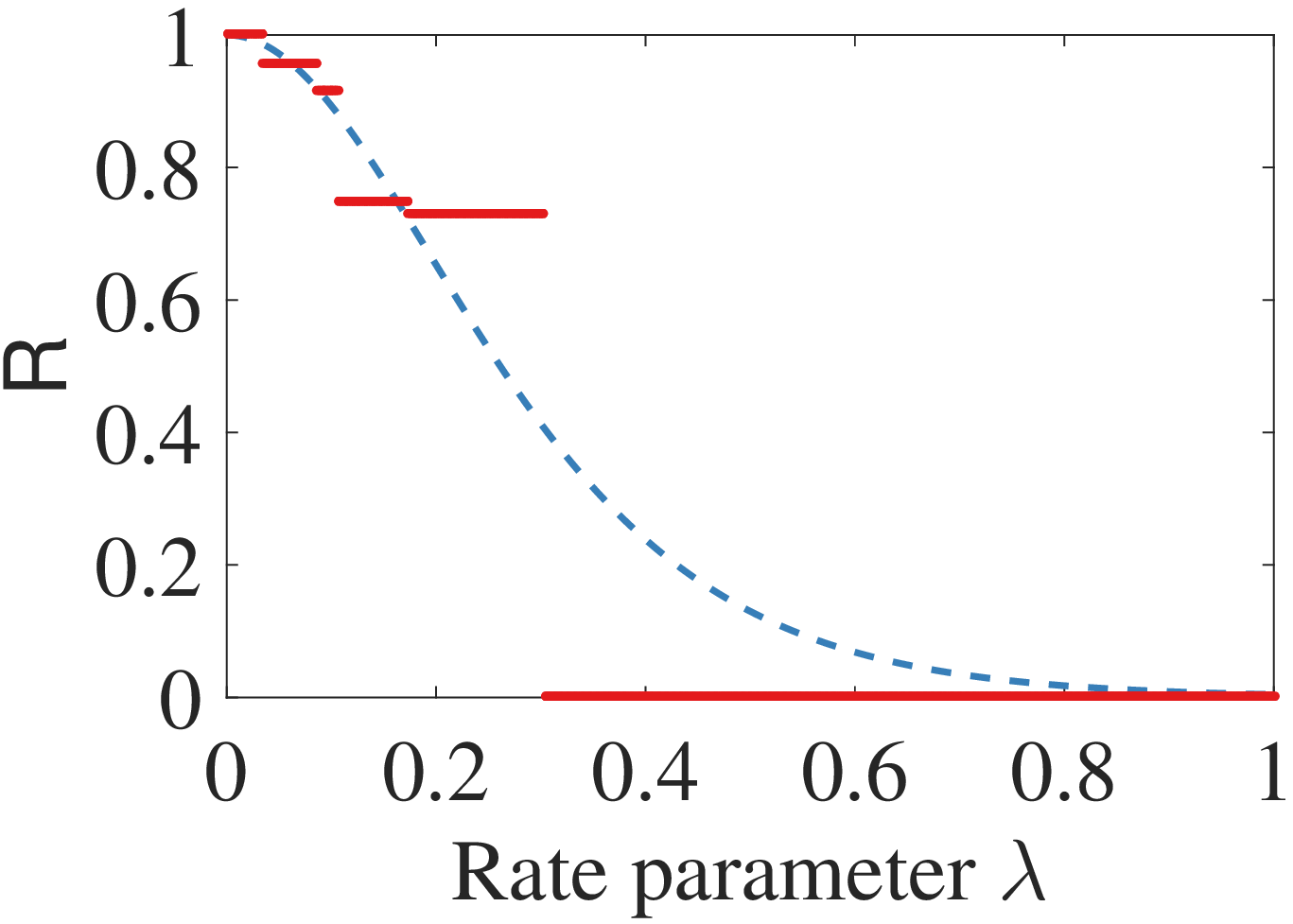}\label{fig:v2f_ghc_b}}
%\subfloat[][]{\!\includegraphics[height=8.5em]{fig_lambda_R_cr.pdf}\label{fig:v2f_ghc_b}}\\[-0.1em]
\caption{(a) Average symbol energy of a continuous MB PMF (dashed lines) and its dyadic approximation (solid lines) divided by that of a uniform PMF $\mathsf{E}_{\text{U}}$, and (b) the corresponding entropy rate. The codewords are taken from Tab.~\ref{tab:examples}(a).}% 
\label{fig:v2f_ghc}
\end{figure}
%%%%%%

If we waive the dyadic constraint on $\mathbb{P}_{\bm{X}}$, the entropy $\mathbb{H}(\bm{X})$ is maximized by $\bm{X}_{\text{MB}}$ that follows an MB distribution
$\mathbb{P}_{\bm{X}_{\text{MB}}}(\bm{x}) := \exp(-\lambda ||\bm{x}||^2) / \sum_{\bm{x}\in\bmcal{\bm{X}}} \exp(-\lambda ||\bm{x}||^2)$.
Note that $\bm{X}_{\text{MB}}$ is a random vector, whereas $X_{\text{MB}}$ in \eqref{eqn:E_gap} is a random scalar.
Since $||\bm{x}||^2 = \sum_{x\in\bm{x}}|x|^2$, we have that $\mathbb{P}_{\bm{X}_{\text{MB}}}(\bm{x}) = \prod_{x\in{\bm{x}}}\mathbb{P}_{X_{\text{MB}}}(x)$.
We also have that $\mathbb{H}(\bm{X}_{\text{MB}}) = v \mathbb{H}(X_{\text{MB}})$.
Therefore, if a dictionary $\bmcal{B}$ can make the codewords $\bm{X}$ follow the distribution $\mathbb{P}_{\bm{X}_{\text{MB}}}$, the code achieves both $\mathsf{E}_\mathcal{C} = \mathbb{E}(|X_{\text{MB}}|^2)$ and $\mathsf{R}_\mathcal{C} = \mathbb{H}(X_{\text{MB}})$.
Indeed, LEC coding \cite{Lempel73} is optimal in the sense that it yields a special instance of the MB distribution, for which the rate parameter is given by $\lambda = \log(\beta)$ with $\beta$ being a root of the characteristic condition $\sum_{\bm{x}\in\bmcal{X}} \beta^{-||\bm{x}||^2} = 1$.
The dashed curves in Fig.~\ref{fig:v2f_ghc}~(a) and (b) show the average symbol energy and the entropy rate of $X_{\text{MB}}$ as a strictly monotonically decreasing function of $\lambda$ (this is true in general, see \cite[Section~5.C]{Bocherer15}), obtained from the codewords of Tab.~\ref{tab:examples}~(a).
Due to the monotonicity, the rate parameter $\lambda$ of the MB PMF $\mathbb{P}_{\bm{X}_\text{MB}}$ that leads to the minimum average symbol energy $\mathbb{E}(|X_{\text{MB}}|^2)$ subject to the rate constraint $\mathbb{H}(X_{\text{MB}}) = \mathsf{R}^*$ can simply be obtained by the \emph{bisection} method.
Once we obtain $\mathbb{P}_{\bm{X}_\text{MB}}$, its optimal dyadic approximate $\mathbb{P}_{\widetilde{\bm{X}}_\text{MB}}$ can be obtained by \emph{Geometric Huffman coding (GHC)}~\cite{Bocherer11}, i.e., $\mathbb{P}_{\widetilde{\bm{X}}_\text{MB}} = \text{GHC} (\mathbb{P}_{\bm{X}_\text{MB}})$, optimal in the sense that it minimizes the KL divergence $\mathbb{D}( \mathbb{P}_{\widetilde{\bm{X}}_\text{MB}} \| \mathbb{P}_{\bm{X}_\text{MB}}) := \sum_{ \bm{x} \in \bmcal{X} } \mathbb{P}_{\widetilde{\bm{X}}_\text{MB}}(\bm{x}) \log_2 \frac{\mathbb{P}_{\widetilde{\bm{X}}_\text{MB}}(\bm{x}) } {\mathbb{P}_{\bm{X}_\text{MB}}(\bm{x})}$.
Notice that, since $\mathcal{X}$ and $v$ completely define the right tree $\bmcal{X}_v$, the PMF $\mathbb{P}_{\bm{X}_\text{MB}}$ can immediately be obtained for an arbitrary target resolution rate $\mathsf{R}^*$, using a series of operations $(\mathcal{X}, v, \mathsf{R}^*) \mapsto \mathbb{P}_{\bm{X}_\text{MB}} \mapsto \mathbb{P}_{\widetilde{\bm{X}}_\text{MB}} \mapsto (\bmcal{B},\bmcal{X}_v)$.
The solid lines in Fig.~\ref{fig:v2f_ghc} show the average symbol energy and resolution rate that can be obtained using this method for $\mathcal{X}_{2\text{-ASK}}$ with $v=3$.
In this example, there are only five non-zero distinct resolution rates generated by the dyadic PMF $\mathbb{P}_{\widetilde{\bm{X}}_\text{MB}}$.
%%%%%%
\begin{figure}[!t]
\centering
\!\!
\includegraphics[width=1.00\linewidth]{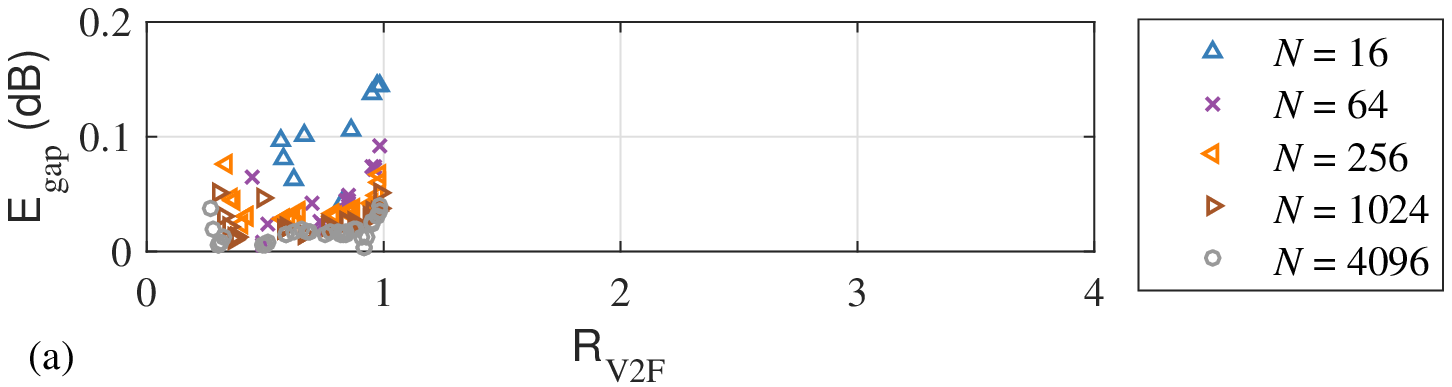}\label{fig:v2f_gap_a}  \\[1.25em]
\includegraphics[width=1.00\linewidth]{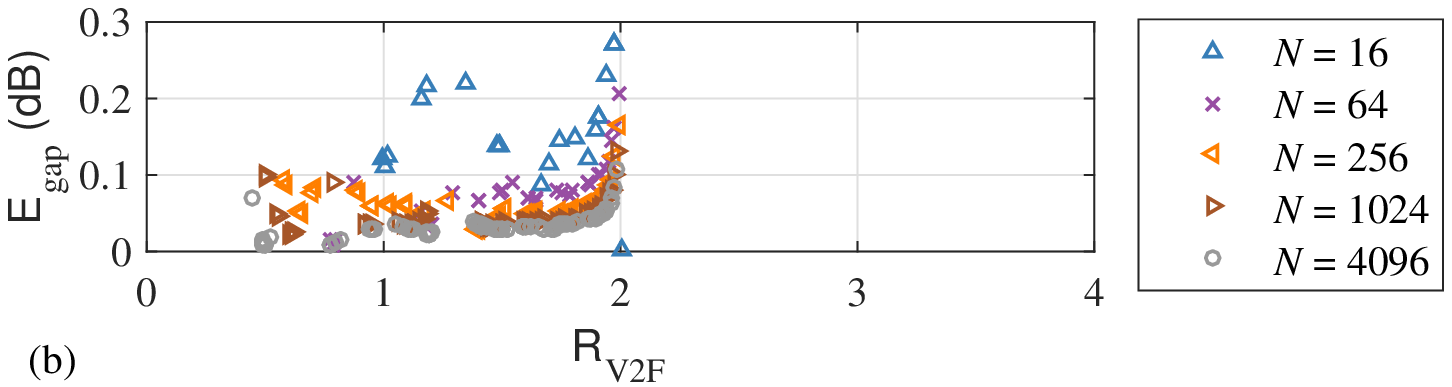}\label{fig:v2f_gap_b}  \\[1.25em]
\includegraphics[width=1.00\linewidth]{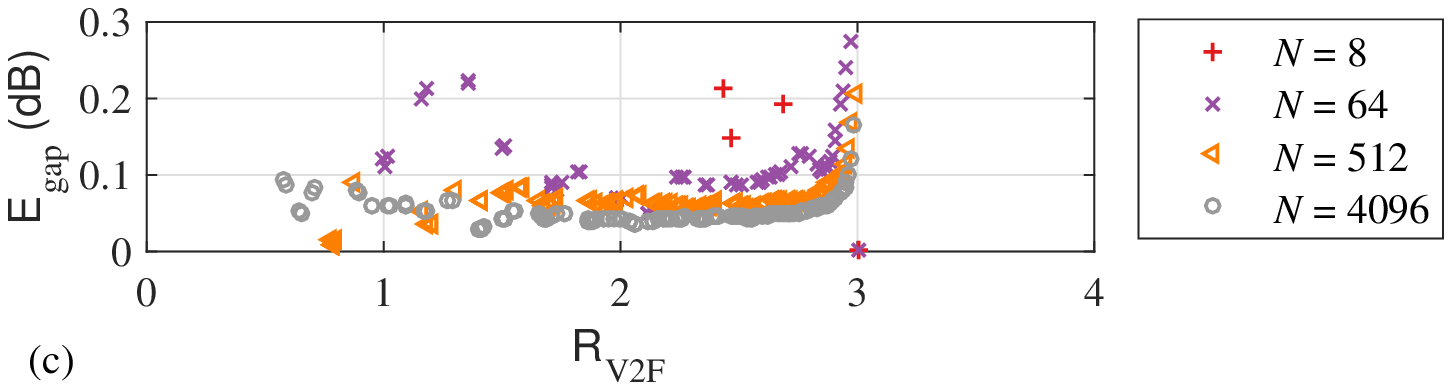}\label{fig:v2f_gap_c}  \\[1.25em]
\includegraphics[width=1.00\linewidth]{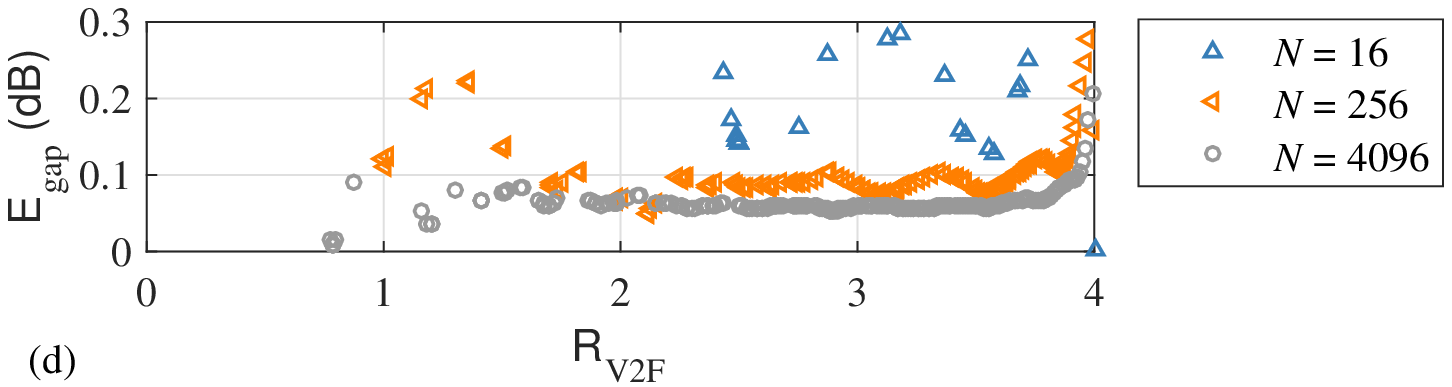}\label{fig:v2f_gap_d}
\caption{Energy gap of V2F codes with $|\mathcal{C}| \leq 4096$ for the (a) 2-ASK, (b) 4-ASK, (c) 8-ASK, and (d) 16-ASK alphabets.}% 
\label{fig:v2f_gap}
\end{figure}
%%%%%%
{However, as $v$ increases, the created codes have much finer granularity of the resolution rate}, as shown in Fig.~\ref{fig:v2f_gap} for up to the 16-ASK alphabet (producing up to the 1024-QAM in the PAS architecture).
Here, we find all V2F codes of the cardinality $|\mathcal{C}| \leq 4096$ for the target resolution rates $\mathsf{R}^* \in \{\Delta,2\Delta,\ldots,\log_2 M\}$ with rate granularity $\Delta = 0.001$.
The codes have less finer granularity of the resolution rate in the lower resolution rate regime, which can be supplemented by allowing variable codeword lengths, as will be discussed in Section~\ref{sec:v2v}.
Figure~\ref{fig:v2f_gap} shows that the constructed V2F codes achieve energy gaps smaller than 0.1~dB across a wide range of resolution rates with a cardinality $|\mathcal{C}| \leq 4096$, or even $|\mathcal{C}| \leq 64$ for $M \leq 4$.

%%%%%%%
%\begin{figure}[!t]
%\centering
%%\includegraphics[width=0.98\linewidth]{fig_v2f_N_R_2ask}\label{fig:v2f_NR_a}  \\[1.5em]
%%\includegraphics[width=0.98\linewidth]{fig_v2f_N_R_4ask}\label{fig:v2f_NR_b}  \\[1.5em]
%%\includegraphics[width=0.98\linewidth]{fig_v2f_N_R_8ask}\label{fig:v2f_NR_c}  \\[1.5em]
%%\includegraphics[width=0.98\linewidth]{fig_v2f_N_R_16ask}\label{fig:v2f_NR_d}
%\includegraphics[width=0.48\linewidth]{fig_v2f_N_R_2ask}\label{fig:v2f_NR_a}  
%\includegraphics[width=0.48\linewidth]{fig_v2f_N_R_4ask}\label{fig:v2f_NR_b}  \\[1.5em]
%\includegraphics[width=0.48\linewidth]{fig_v2f_N_R_8ask}\label{fig:v2f_NR_c}  
%\includegraphics[width=0.48\linewidth]{fig_v2f_N_R_16ask}\label{fig:v2f_NR_d}
%\caption{Entropy rates of V2F codes for $N \leq 4096$, with the (a) 2-ASK, (b) 4-ASK, (c) 8-ASK, and (d) 16-ASK alphabets.}% 
%\label{fig:v2f_NR}
%\end{figure}
%%%%%%%
%

%%%%%%%%%%%%%%%%%%%%%%%%%%%%%%%%%%%%%%%%%%%%%%%%%%%%%%%%%%%%%%%%%%%%%%%%%
\subsection{F2V Codes}

For F2V codes, we consider a balanced binary left tree representing a dictionary $\bmcal{B}$ with a fixed information word length $l(\bm{b}) = u$ for all $\bm{b}\in \bmcal{B}$ such that $|\mathcal{C}| = 2^u$.
Then, the dictionary $\bmcal{B}_u$ can be immediately obtained by lexicographical ordering of $\mathcal{B}^u$, and all the information words $\bm{b}\in\bmcal{B}_u$ are parsed to the codewords $\bm{x} \in \bmcal{X}$ with the equal probability $\mathbb{P}_{\bm{B}}(\bm{b}) = 2^{-u}$.
Let us define the \emph{sum depth} and \emph{sum energy} of a right tree as $l(\bmcal{X}) := \sum_{\bm{x}\in\bmcal{X}} l(\bm{x})$ and $||\bmcal{X}||^2 := \sum_{\bm{x}\in\bmcal{X}} ||\bm{x}||^2$, respectively.
Then, the average symbol energy and the resolution rate of the code are calculated, respectively, as
\begin{align*} 
\mathsf{E}_\mathcal{C}
= \frac{\mathbb{E}(||\bm{X}||^2)}{\mathbb{E}(l(\bm{X}))}
= \frac{\sum_{\bm{x}\in\bmcal{X}}2^{-u}||\bm{x}||^2}{\sum_{\bm{x}\in\bmcal{X}}2^{-u}l(\bm{x})}
= \frac{||\bmcal{X}||^2}{l(\bmcal{X})}, \eqnum \label{eqn:E_f2v}
\end{align*}
and
\begin{align*}
\mathsf{R}_\mathcal{C}
= \frac{\mathbb{E}(l(\bm{B}))}{\mathbb{E}(l(\bm{X}))}
= \frac{u}{\sum_{\bm{x}\in\bmcal{X}}2^{-u}l(\bm{x})}
= \frac{2^{u} u}{l(\bmcal{X})}.	\eqnum \label{eqn:R_f2v}
\end{align*}

%%%%%%%
\begin{figure}[!t]
\centering
\subfloat[][]{\!\includegraphics[width=0.97\linewidth]{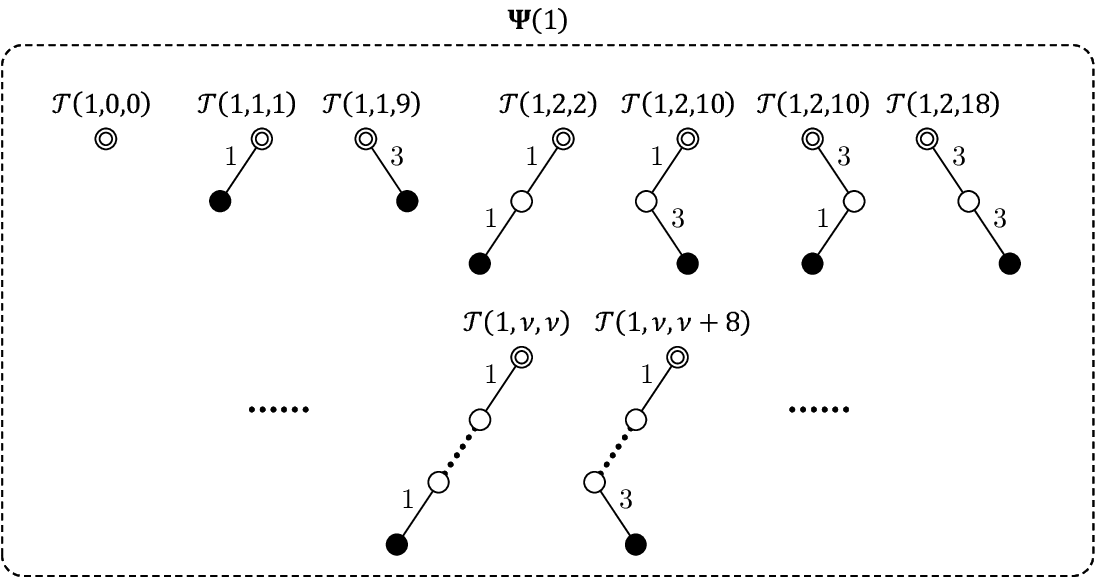}\label{fig:Lplus_trees_a}}\\[0.5em]
\subfloat[][]{\!\includegraphics[width=0.97\linewidth]{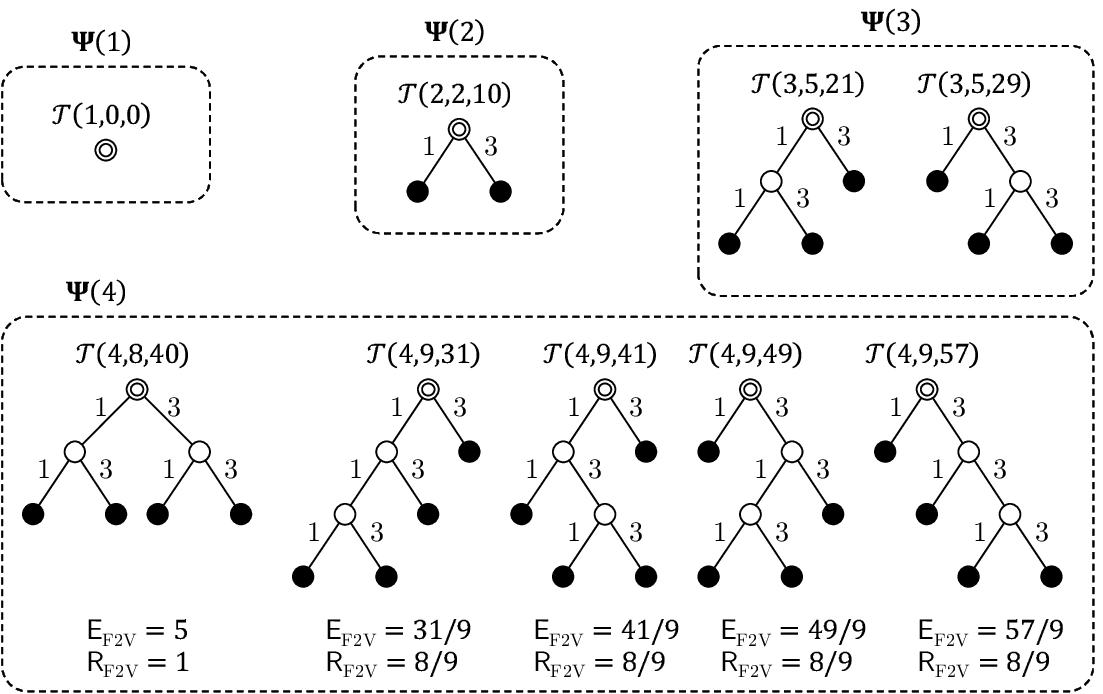}\label{fig:Lplus_trees_b}}
%\subfloat[][]{\!\includegraphics[width=0.46\linewidth]{fig_1plus_trees_cropped.pdf}\label{fig:Lplus_trees_a}} \qquad
%\subfloat[][]{\!\includegraphics[width=0.47\linewidth]{fig_2plus_trees_cropped.pdf}\label{fig:Lplus_trees_b}} \\[-0.5em]
\caption{All $L^+$-trees in $\bm{\Psi}(N)$ with the 2-ASK alphabet for (a) $L = 1, N = 1$, and (b) $L=2, N \leq 4$.}% 
\label{fig:Lplus_trees}
\end{figure}
%%%%%%%

We do not know of any existing method that can solve~\eqref{eqn:optimize_E} with \eqref{eqn:E_f2v} and \eqref{eqn:R_f2v} for an arbitrary desired $\mathsf{R}^*$.
In this paper, therefore, instead of directly solving the problem for a particular $\mathsf{R}^*$, we try to identify a set of trees that produces all possible sum depths $l(\bmcal{X})$ under a tree size constraint $|\bmcal{X}|=u$, hence realizing all possible $\mathsf{R}_\mathcal{C}$ according to~\eqref{eqn:R_f2v}.
The constructed trees are optimal in the sense that they have the minimum sum energy $||\bmcal{X}||^2$ among all $2^+$-trees with the same sum depth~$l(\bmcal{X})$, where $L^+$-tree is a tree in which every node except leaves has a degree not smaller than $L$.
This further implies that the tree achieves the minimum average symbol energy $\mathsf{E}_\mathcal{C}$ for the given $l(\bmcal{X})$ and $\mathsf{R}_\mathcal{C}$ due to \eqref{eqn:E_f2v} and \eqref{eqn:R_f2v}.
Indeed, any $1^+$-tree is allowed for a right tree; however, we restrict the type of trees to $2^+$-trees to avoid infinite tree expansion (see Fig.~\ref{fig:Lplus_trees}~(a) for example).
In Fig.~\ref{fig:Lplus_trees}, and throughout the paper, $\mathcal{T}(N)$ denotes a right tree that has $N$ leaves, i.e., $|\bmcal{X}|=N$, $\mathcal{T}(N,\nu)$ is a tree $\mathcal{T}(N)$ whose sum depth $l(\bmcal{X})=v$, and $\mathcal{T}(N,\nu,\omega)$ is a tree $\mathcal{T}(N,\nu)$ whose sum energy $||\bmcal{X}||^2=\omega$.
Also, let $\bm{\Psi}(N)$ and $\bm{\Psi}(N,\nu)$ be the sets of all trees $\mathcal{T}(N)$ and $\mathcal{T}(N,\nu)$, respectively.
Then, by \eqref{eqn:E_f2v} and \eqref{eqn:R_f2v}, all trees in $\bm{\Psi}(N,\nu)$ lead to the same $\mathsf{R}_\mathcal{C}$ but not necessarily the same $\mathsf{E}_\mathcal{C}$. 
Since a distinct sum depth $\nu$ generates a distinct $\mathsf{R}_\mathcal{C} = 2^u u/\nu$ for the fixed $u$, we have as many optimal trees as the number of distinct sum depths of trees in $\bm{\Psi}(N)$, where an optimal tree with sum depth $\nu$ is defined as $\mathcal{T}^* (N,\nu) = \argmin_{\mathcal{T} (N,\nu) \in \bm{\Psi}(N,\nu)} ||\bmcal{X}||^2$.
A brute-force search of $\mathcal{T}^* (N,\nu)$ in $\bm{\Psi}(N)$ requires exponential time in $N$.
There are known problems that are isomorphic to the problem of counting the number of all trees in $\bm{\Psi}(N)$, including the parenthesizations counting problem \cite[Ch.~15.2]{Cormen01}.
For example, the trees in $\bm{\Psi}(4)$ of Fig.~\ref{fig:Lplus_trees}~(b) have isomorphic representations of $((13)(13))$, $(((13)3)3)$, $((1(13))3)$, $(1((13)3))$, and $(1(1(13)))$ in order, in which the edges connecting two siblings are parenthesized together in a recursive manner from the largest depth of the tree.
The number of parenthesizations for $\bm{\Psi}(N)$ is $C_{N-1}$, with $C_{N}$ being the \emph{Catalan number} defined as $C_{N} := \frac{(2N)!}{N!(N+1)!}$ \cite{Stanley15,Dershowitz80}, which grows as $\Omega(\frac{4^N}{N^{3/2}})$ \cite[p.~333]{Cormen01}.

In order to reduce the search space, we take a \emph{dynamic programming} approach, as in \cite{Golin98,Golin08}, to find a prefix-free code for a known codeword PMF.
%%%%%%
\begin{figure}[!t]
\centering
\includegraphics[width=0.95\linewidth]{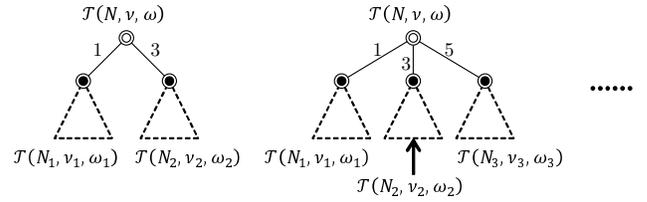}  \label{fig:v2f_NR_a}
\label{fig:v2f_NR_a}
\caption{A large tree constructed from smaller sub-trees.}% 
\label{fig:tree_construction}
\end{figure}
%%%%%%
We begin from a trivial size-1 tree that has only a root, then construct larger trees by appending smaller trees to the root, as shown in Fig.~\ref{fig:tree_construction}.
A tree $\mathcal{T}(N,\nu,\omega)$ formed by appending sub-trees $\mathcal{T}(N_j,\nu_j,\omega_j)$, $j=1,\ldots,J$ satisfies the following relations:
\begin{align*}
N &= \sum_{j=1}^{J} N_j, 												\eqnum \label{eqn:update_N} \\
\nu &= \sum_{j=1}^{J} N_j + \sum_{j=1}^{J} \nu_j = N + \sum_{j=1}^{J} \nu_j, 	\eqnum \label{eqn:update_nu} \\
\omega &= \sum_{j=1}^{J} (2j-1)^2 N_j + \sum_{j=1}^{J} \omega_j, 			\eqnum \label{eqn:update_omega}
\end{align*}
where the last two equations hold since the $j$-th edge from the root should be taken into account $N_j$ times to calculate sum depth and sum energy of the tree.
For example, due to~\eqref{eqn:update_N}, a tree of $N=4$ with the $4$-ASK alphabet (hence $2 \leq J \leq 4$) can be constructed from $J$-tuples: $(N_1, N_2) = (1,3), (2,2), (3,1)$ for $J=2$, $(N_1, N_2, N_3) = (1,1,2), (1,2,1), (2,1,1)$ for $J=3$, and $(N_1, N_2, N_3, N_4) = (1,1,1,1)$ for $J=4$. 
Here, we do not need to use all sub-trees of size $<N$ to construct a set of size-$N$ trees, since we have an \emph{optimal sub-substructure} property:
\begin{thm}\label{thm_opt_subtree}
An optimal tree $\mathcal{T}^* (N)$ contains only optimal sub-trees $\mathcal{T}^* (n < N)$ in it.% with $n\in\{1,\ldots,N-1\}$.
%Given the sets $\bm{\Psi}^* (n)$ of optimal trees $\mathcal{T}^* (n)$ for all $n=1,\ldots,N-1$, We can efficiently build a set $\bm{\Psi}^* (N)$ of optimal trees $\mathcal{T}^* (N)$ by  i.e., tree with the smallest energies in $\bm{\Psi}^* (N)$ can be constructed from 
\end{thm}
%%%
%%%
\begin{proof}%[\unskip\nopunct]
Suppose that an optimal tree $\mathcal{T}^* (N,\nu,\omega^*)$ has a $j$-th sub-tree $\mathcal{T} (N_j,\nu_j,\omega_j)$.
If the sub-tree $\mathcal{T} (N_j,\nu_j,\omega_j)$ is not an optimal tree, then there exists a sub-tree $\mathcal{T}^* (N_j,\nu_j,\omega_j^*)$ with $\omega_j^* < \omega_j$.
By replacing the sub-tree $\mathcal{T} (N_j,\nu_j,\omega_j)$ with $\mathcal{T}^* (N_j,\nu_j,\omega_j^*)$ from the tree $\mathcal{T}^* (N,\nu,\omega^*)$, we can obtain a new tree $\mathcal{T}' (N,\nu,\omega')$ that has a smaller sum energy than the optimal tree, i.e, with $\omega'<\omega^*$.
This is a contradiction, hence an optimal tree consists of only optimal sub-trees.
\end{proof}
%%%

It can be easily seen that the height of $\mathcal{T}(N)$ is minimized when its leaves have maximally uniform depths.
The depths can differ from each other at most by one, and the minimum height is given by $\lceil \log_{M}N \rceil$ for the $M$-ASK alphabet.
%It is obvious that the height of a tree $\mathcal{T}(N)$ is lower-bounded by $\lceil \log_{M}N \rceil$, achieved when the leaves have maximally uniform depths.
The sum depth of a size-$N$ minimum-height tree is calculated as
$\nu_{\text{min}}	:= (N - M^{\lfloor \log_M N \rfloor}) (\lfloor \log_M N \rfloor + 1) 
 + M^{\lfloor \log_M N \rfloor} \lfloor \log_M N \rfloor$, which is equal to the smallest sum depth of all trees in $\bm{\Psi} (N)$.
If $N$ is a positive integer power of $M$, the minimum sum depth can be simplified as $\nu_{\text{min}} = N \log_M N$.
On the other hand, if there are only two nodes at every depth of a tree (see, e.g., the last tree of Fig.~\ref{fig:Lplus_trees}~(b)), the tree has height $N-1$ that is the largest of all trees in $\bm{\Psi}(N)$.
The sum depth in this case is calculated as $\nu_{\text{max}} 	= \sum_{v=1}^{N-1} v + (N-1) = \frac{(N+2)(N-1)}{2}$.			After some manipulation, it can be seen that the sum depth of a maximum-height tree is also the maximum sum depth of all trees in $\bm{\Psi} (N)$. %\emph{(((Proof needed?)))}
Therefore, if we store only one optimal tree to $\bm{\Psi}^* (N)$ for each $\nu$, the size of $\bm{\Psi}^* (N)$ is upper-bounded by $\nu_{\text{max}} - \nu_{\text{min}} + 1 = \frac{N^2}{2} + O(N \log_M N)$; i.e., $|\bm{\Psi}^*(N)|$ grows as $O(N^2)$.
Since there are at most $N$ choices of the optimal sub-tree sets (ranging from an empty set to $\bm{\Psi}^*(N-1)$) for each of the $M$ edges of the root for enumerating all trees of size $N$ (using the optimal sub-tree property), the number of all choices to build $\bm{\Psi}(N)$ is of the complexity $\Omega(N^M)$.
As aforementioned, each of the sub-tree sets is of size at most $\Omega(N^2)$, hence we have the search space of size $\Omega(N^{M+2})$ to identify $\bm{\Psi}^*(N)$.
The search time is polynomial in $N$, and exponential in $M$, indicating the intractability of constructing F2V codes for a large alphabet.

To further reduce the search space, we can exploit the fact that a larger sub-tree is appended as far left as possible from the root of an optimal tree; i.e., 
\begin{thm}
\label{thm_opt_subtree2}
Let $\mathcal{T}_j$ for $j=1,\ldots,J$ constitute the $j$-th sub-tree of an optimal tree $\mathcal{T}^*$ such that the root of $\mathcal{T}_j$ is the $j$-th child of the root of $\mathcal{T}^*$, then the sub-trees satisfy $|\mathcal{T}| > |\mathcal{T}_1| \geq \ldots \geq |\mathcal{T}_J|$.
\end{thm}
%%%
%%%
\begin{proof}%[\unskip\nopunct]
Let $\mathcal{T} (N_i,\nu_i)$ and $\mathcal{T} (N_j,\nu_j)$ denote two of the sub-trees of an optimal tree $\mathcal{T}^* (N,\nu)$ with $i < j$, and assume $N_i < N_j$.
Then, by exchanging the two sub-trees $\mathcal{T} (N_i,\nu_i)$ and $\mathcal{T} (N_j,\nu_j)$, we can construct another tree $\mathcal{T}' (N,\nu)$ of the same size $N$ and the same sum depth $\nu$ due to \eqref{eqn:update_N} and \eqref{eqn:update_nu}, which has a smaller sum energy than $\mathcal{T}^* (N,\nu)$ by~\eqref{eqn:update_omega}.
This is contradiction, hence an optimal tree must have $N_i \geq N_j$ for any $i > j$.
\end{proof}
\noindent For example, some trees in $\bm{\Psi}(4)$ can be generated from three sub-trees that have $(N_1,N_2,N_3)=(2,1,1)$, $(1,2,1)$, or $(1,1,2)$, but we can discard the latter two triplets by Theorem~\ref{thm_opt_subtree2}.

%%%%%%
%\begin{figure}[!t]
%\centering
%%\includegraphics[width=1.00\linewidth]{fig_f2v_N_R_2ask}\label{fig:f2v_NR_a}  \\[1.5em]
%%\includegraphics[width=1.00\linewidth]{fig_f2v_N_R_4ask}\label{fig:f2v_NR_b}
%\includegraphics[width=0.48\linewidth]{fig_f2v_N_R_2ask}\label{fig:f2v_NR_a}
%\includegraphics[width=0.48\linewidth]{fig_f2v_N_R_4ask}\label{fig:f2v_NR_b}
%\caption{Entropy rates of all F2V codes for $N \leq 64$, with the (a) 2-ASK, and (b) 4-ASK alphabets.}%
%\label{fig:f2v_NR}
%\end{figure}
%%%%%%

%%%%%%
\begin{figure}[!t]
\centering
\includegraphics[width=0.97\linewidth]{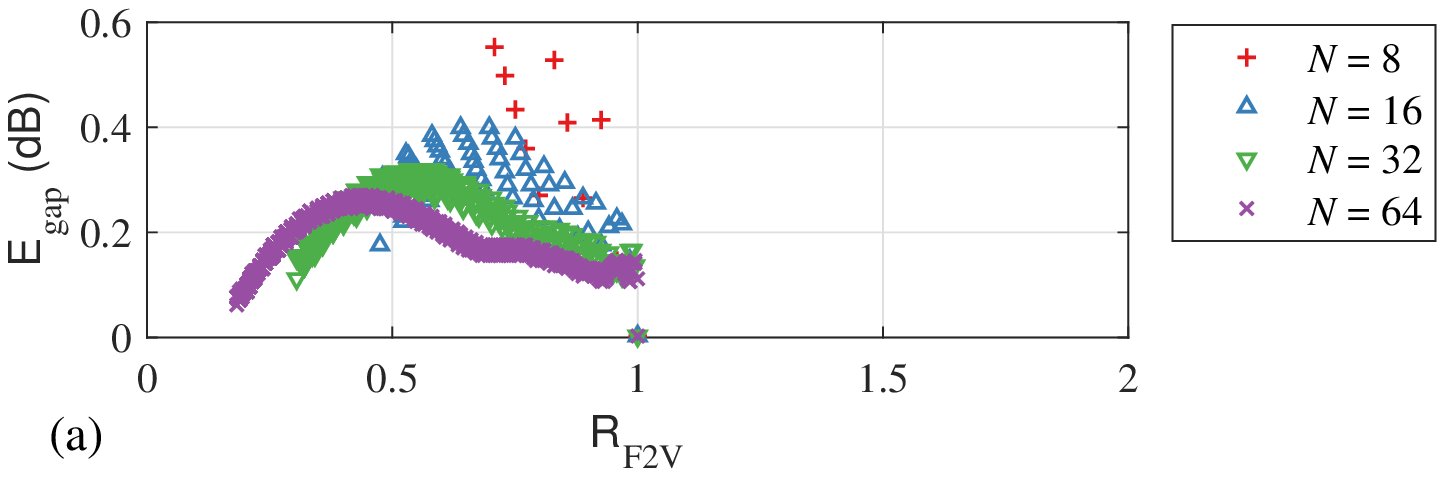}\label{fig:f2v_gap_a} \\[1.2em]
\includegraphics[width=0.97\linewidth]{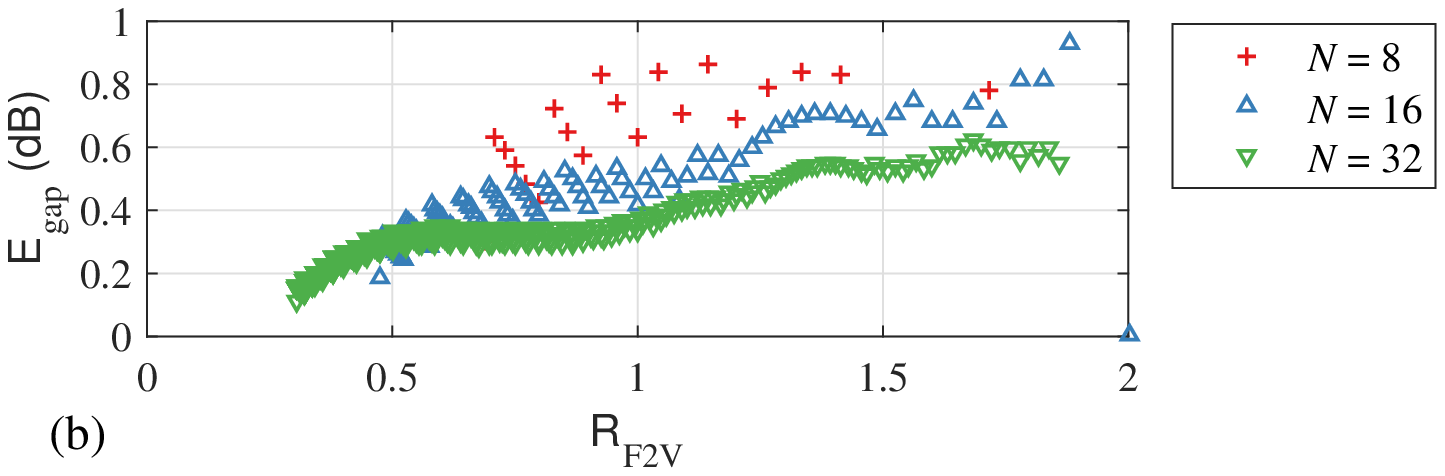}\label{fig:f2v_gap_b} 
\caption{Energy gap of all F2V codes (a) with $|\mathcal{C}| \leq 64$ for the 2-ASK alphabet, and (b) with $|\mathcal{C}| \leq 32$ for the 4-ASK alphabet.}%
\label{fig:f2v_gap}
\end{figure}
%%%%%%

We construct one optimal F2V code for each of the resolution rates that can be created with $N \leq 64$ for $\mathcal{X}_{\text{2-ASK}}$ and $\mathcal{X}_{\text{4-ASK}}$.
As shown in Fig.~\ref{fig:f2v_gap}, F2V codes offer finer granularity of the resolution rates than V2F codes with the same cardinality, albeit with a slightly larger energy gap.
Therefore, a natural consequence is to construct V2V codes to exploit the complementary merits of V2F and F2V codes, as will be illustrated henceforth.

%%%%%%%%%%%%%%%%%%%%%%%%%%%%%%%%%%%%%%%%%%%%%%%%%%%%%%%%%%%%%%%%%%%%%%%%%
\subsection{V2V Codes}
\label{sec:v2v}

We readily have a set $\bm{\Psi}_{\text{R}}^*(N)$ of optimal right trees, representing optimal F2V codes.
From this, without imposing any length constraints on information words and codewords, we can enumerate near-optimal V2V codes of the cardinality $|\mathcal{C}| = N$ in the following manner:
\begin{enumerate}
  \item For each of the desired resolution rates $\mathsf{R}^* = j \Delta_{\mathsf{R}}$, where $j=1,\ldots,J$ for some integer $J$ and the \emph{rate granularity} $\Delta_{\mathsf{R}} = \frac{\log_2 M}{J}$, and for each of the optimal right trees in $\bm{\Psi}_{\text{R}}^*$, identify the MB PMF $\mathbb{P}_{\bm{X}_\text{MB}}$ that fulfills $\mathbb{H}(X_{\text{MB}}) = \mathsf{R}^*$.
  \item For every $\mathbb{P}_{\bm{X}_\text{MB}}$ obtained in Step~1, construct a set  of optimal left trees $\bm{\Psi}_{\text{L}}^* $ that realizes $\mathbb{P}_{\widetilde{\bm{X}}_\text{MB}} = \text{GHC}(\mathbb{P}_{\bm{X}_\text{MB}})$. The trees in $\bm{\Psi}_{\text{L}}^* $ and $\bm{\Psi}_{\text{R}}^* $ have a one-to-one correspondence.
 \item For each of $\mathsf{R}^* = k \Delta_{\mathsf{R}}$, choose a pair of the left and right trees in $\bm{\Psi}_{\text{L}}^* $ and $\bm{\Psi}_{\text{R}}^*$ that yield the minimum $\mathsf{E}_\mathcal{C}$ with a rate discrepancy $|\mathsf{R}_\mathcal{C}-\mathsf{R}^*| < \delta$ for a small $\delta$.
\end{enumerate}

\noindent Indeed, an optimal V2V code does not necessarily consist of an optimal right tree, hence an exhaustive search of the left tree should be performed over \emph{all} right trees in $\bm{\Psi}_{\text{R}}$.
However, due to the exponential growth of $|\bm{\Psi}_{\text{R}}|$ in $N$, we restrict the search space to optimal right trees in $\bm{\Psi}_{\text{R}}^*$; under this constraint, the constructed V2V codes show surprisingly good performance with a very small $N \leq 32$.

Assume that a right tree $\mathcal{T}^*(N)$ has been chosen from $\bm{\Psi}_{\text{R}}^*(N)$, hence $l(\bm{x})$ and $||\bm{x}||^2$ are given.
Then, we have from~\eqref{eqn:R} that 
\begin{align}
\mathsf{R}_\mathcal{C} = (-\sum_{n=1}^{N} p_n \log_2 p_n) / (\sum_{n=1}^{N} p_n l(\bm{x}_n) ) \geq \mathsf{R}^* \\\Leftrightarrow \quad {-\sum_{n=1}^{N} p_n \log_2 (p_n / 2^{- \mathsf{R}^* l(\bm{x}_n)}) } \geq 0.
\end{align}
Let $ q_n := 2^{- \mathsf{R}^* l(\bm{x}_n)}$, $\bm{q} = [q_1,\ldots,q_N]^T$.
Then, by definition of $\mathbb{D}(\bm{p} \| \bm{q})$, we have an equivalence relation
\begin{gather*} 
\mathsf{R}_\mathcal{C} \geq \mathsf{R}^*  \quad  \Leftrightarrow  \quad  -\mathbb{D}(\bm{p} \| \bm{q}) \geq 0.	\eqnum  \label{eqn:KL_v2v}
\end{gather*}
If we waive the dyadic constraint on $\bm{p}$, the optimization problem~\eqref{eqn:optimize_E} translates to
%\begin{align*}
%\begin{array}{ll}
%\underset{\bm{p}}{\text{minimize}} & \mathsf{E}(\bmcal{X}) = \frac { \sum_{n=1}^{N} p_n ||\bm{x}_n||^2 } { \sum_{n=1}^{N} p_n l(\bm{x}_n) }  \eqnum \label{eqn:opt_v2v_minE1}		\\
%\text{subject to} &  -\mathbb{D}(\bm{p} \| \bm{q}) \geq 0,\eqnum \label{eqn:opt_v2v_st1} \\
%%			 	& \sum_{n=1}^{N} p_n = 1. \eqnum \label{eqn:opt_v2v_st2}
%\end{array}
%\end{align*}
\begin{align*}
& \underset{\bm{p}}{\text{minimize}} & &\!\!\!\!\!\!\!\!\!\!\!\! \mathsf{E}_\mathcal{C} = \frac { \mathbb{E}(||\bm{X}||^2) } { \mathbb{E}(l(\bm{X})) } = \frac { \sum_{n=1}^{N} p_n ||\bm{x}_n||^2 } { \sum_{n=1}^{N} p_n l(\bm{x}_n) }  \eqnum \label{eqn:opt_v2v_minE1}		\\
& \text{subject to} 					& &\!\!\!\!\!\!\!\!\!\! -\mathbb{D}(\bm{p} \| \bm{q}) \geq 0,\eqnum \label{eqn:opt_v2v_st1} \\
&									& &\!\!\!\!\!\!\!\!\!\! \sum_{n=1}^{N} p_n = 1. \eqnum \label{eqn:opt_v2v_st2}
\end{align*}
%\begin{align*}
%& \underset{\bm{p}}{\text{minimize}} & &\!\!\!\!\!\!\!\!\!\!\!\!\!\!\!\!\!\! \mathsf{E} = \frac { \sum_{n=1}^{N} p_n w_n } { \sum_{n=1}^{N} p_n v_n }  \eqnum \label{eqn:opt_v2v_minE1}		\\
%& \text{subject to} 					& &\!\!\!\!\!\!\!\!\!\!\!\!\!\!\!\!\!\! -\mathbb{D}(\bm{p} \| \bm{q}) \geq 0,\eqnum \label{eqn:opt_v2v_st1} \\
%&									& &\!\!\!\!\!\!\!\!\!\!\!\!\!\!\!\!\!\! \sum_{n=1}^{N} p_n = 1. \eqnum \label{eqn:opt_v2v_st2}
%\end{align*}
Let $\mathsf{E}^*$ denote the minimum energy found by solving \eqref{eqn:opt_v2v_minE1}-\eqref{eqn:opt_v2v_st2}.
Since the average symbol energy of any code $\mathcal{C}$ fulfills
\begin{align*} 
\mathsf{E}_\mathcal{C} = \frac { \mathbb{E}(||\bm{X}||^2) } { \mathbb{E}(l(\bm{X})) }  \geq \mathsf{E}^* \quad
\Leftrightarrow	\quad \mathbb{E}(||\bm{X}||^2)  - \mathsf{E}^* \cdot \mathbb{E}(l(\bm{X})) \geq 0, \eqnum \label{eqn:opt_v2v_nec_suf}
\end{align*}
where the equality in the right-hand side holds if and only if $\mathsf{E}_\mathcal{C} = \mathsf{E}^*$, the optimal PMF is a solution to a convex optimization problem
\begin{align}
\bm{p}^* = \underset{\bm{p}}{\text{argmin}} \: \left[\mathbb{E}(||\bm{X}||^2) - \mathsf{E}^* \cdot \mathbb{E}(l(\bm{X})) \right].
\end{align}
%%%%%%%%%%%%%%%
\begin{algorithm}[t!]
\label{alg:v2v}
\caption{Finding an optimal codeword PMF $\bm{p}^*$ for $\mathsf{R}^*$}
\begin{algorithmic}[1]
\renewcommand{\algorithmicrequire}{\textbf{Input:}}
\renewcommand{\algorithmicensure}{\textbf{Output:}}
\REQUIRE $N,\bm{v,w,q}$
\ENSURE  $\bm{p}^*$
\\ \textit{Initialization}: \\
$\ell\leftarrow0$, $\bm{p}^{(0)} \leftarrow \bm{q}/\sum_{n=1}^{N} q_n$, ${\mathsf{E}}_\mathcal{C}^{(0)} \leftarrow \mathbb{E}_{\bm{p}^{(0)}} (||\bm{X}||^2)/ \mathbb{E}_{\bm{p}^{(0)}} (l(\bm{X}))$.
\REPEAT
	\STATE $ \ell \leftarrow \ell + 1$
	\STATE $\bm{p}^{(\ell)} \leftarrow \underset{\bm{p}}{\text{argmin}} \: \left[\mathbb{E}(||\bm{X}||^2) - {\mathsf{E}}_\mathcal{C}^{(\ell-1)} \cdot \mathbb{E}(l(\bm{X})) \right]$\\
	$\qquad\quad\:$subject to $-\mathbb{D}(\bm{p} \| \bm{q}) \geq 0$, $\sum_{n=1}^{N} p_n = 1$.
	\STATE ${\mathsf{E}}_\mathcal{C}^{(\ell)} \leftarrow  \mathbb{E}_{\bm{p}^{(\ell)}} (||\bm{X}||^2)/ \mathbb{E}_{\bm{p}^{(\ell)}} (l(\bm{X}))$.
\UNTIL{$\mathsf{E}_\mathcal{C}^{(\ell-1)} - \mathsf{E}_\mathcal{C}^{(\ell)} < \varepsilon$}
\RETURN $\bm{p}^* \leftarrow \bm{p}^{(\ell)}$
\end{algorithmic}
\end{algorithm}
%%%%%%%%%%%%%%%
%%%%%%%%%%%%%%%%
%\begin{algorithm}[t!]
%\label{alg:v2v}
%\caption{Finding an optimal codeword PMF $\bm{p}^*$ for $\mathsf{R}^*$}
%\begin{algorithmic}[1]
%\renewcommand{\algorithmicrequire}{\textbf{Input:}}
%\renewcommand{\algorithmicensure}{\textbf{Output:}}
%\REQUIRE $N,\bm{v,w,q}$
%\ENSURE  $\bm{p}^*$
%\\ \textit{Initialization}: \\
%$\ell\leftarrow0$, $\bm{p}^{(\ell)} \leftarrow \bm{q}/\sum_{n=1}^{N} q_n$, ${\mathsf{E}}^{(\ell)} \leftarrow \mathbb{E}_{\bm{p}^{(\ell)}} (W)/ \mathbb{E}_{\bm{p}^{(\ell)}} (V)$.
%\REPEAT
%	\STATE $ \ell \leftarrow \ell + 1$
%	\STATE $\bm{p}^{(\ell)} \leftarrow \underset{\bm{p}}{\text{argmin}} \: \left[\mathbb{E}(W) - {\mathsf{E}}^{(\ell-1)} \cdot \mathbb{E}(V) \right]$\\
%	$\qquad\quad\:$subject to $-\mathbb{D}(\bm{p} \| \bm{q}) \geq 0$, $\sum_{n=1}^{N} p_n = 1$.
%	\STATE ${\mathsf{E}}^{(\ell)} \leftarrow  \mathbb{E}_{\bm{p}^{(\ell)}} (W)/ \mathbb{E}_{\bm{p}^{(\ell)}} (V)$.
%\UNTIL{$\mathsf{E}^{(\ell-1)} - \mathsf{E}^{(\ell)}  < \varepsilon$}
%\RETURN $\bm{p}^* \leftarrow \bm{p}^{(\ell)}$
%\end{algorithmic}
%\end{algorithm}
%%%%%%%%%%%%%%%%
Also, the constraints in \eqref{eqn:opt_v2v_st1} and \eqref{eqn:opt_v2v_st2} are concave\cite[Th.~2.7.2]{Cover06} and affine, respectively, hence the problem can efficiently be solved by a convex optimization solver such as the CVX \cite{cvx,Grant08}, on condition that $\mathsf{E}^*$ is known.
Since we do not know $\mathsf{E}^*$, we can make an initial guess $\mathsf{E}_\mathcal{C}^{(0)} = \mathbb{E}(||\bm{X}||^2) / \mathbb{E}(l(\bm{X}))$ using a PMF $\bm{p} = \bm{q} / \sum_{n=1}^{N} q_n$, then attempt to reduce the error between the estimate ${\mathsf{E}}_\mathcal{C}^{(\ell)}$ at iteration $\ell$ and the true minimum energy $\mathsf{E}^*$ in an iterative manner.
The initial PMF maximizes $-\mathbb{D}(\bm{p} \| \bm{q})$ in an attempt to fulfill the rate condition~\eqref{eqn:KL_v2v}, where the maximization of $-\mathbb{D}(\bm{p} \| \bm{q})$ by the initial guess $\bm{p}$ can be proven by using a Lagrangian $\mathbb{L}(\bm{p},\lambda) = -\mathbb{D}(\bm{p} \| \bm{q}) + \lambda(\sum_{n=1}^{N} p_n - 1)$.
Since $\bm{p}$ is strictly feasible, the Slater's condition is satisfied for the problem of maximizing the concave function $-\mathbb{D}(\bm{p} \| \bm{q})$, hence strong duality holds\cite[Ch.~5.2.3]{Boyd04}.
It follows that the \textrm{Karush-Kuhn-Tucker} (KKT) conditions \cite[Ch.~5.5.3]{Boyd04} are necessary and sufficient to maximize  $-\mathbb{D}(\bm{p} \| \bm{q})$.
Namely, a PMF $\bm{p}$ maximizes $-\mathbb{D}(\bm{p} \| \bm{q})$ if and only if the gradient of the Lagrangian vanishes at $\bm{p}$; i.e.,
\begin{align*}
&\frac{\partial \mathbb{L}(\bm{p},\lambda)}{\partial p_n} = -\log_2 \frac{p_n}{q_n} - \log_2 e + \lambda = 0 \\
&\Leftrightarrow 	\quad \log_2 \frac{p_n}{q_n} = \lambda - \log_2 e\\
&\Leftrightarrow 	\quad p_n = 2^{\lambda - \log_2 e} q_n = r q_n \text{ for some constant }r.
\end{align*}
Since $\bm{p}$ is a PMF, we have that $r = 1/\sum_{n=1}^{N} q_n$ in the last equation, hence the initial guess $\bm{p}$ maximizes $-\mathbb{D}(\bm{p} \| \bm{q})$.
If this $\bm{p}$ does not fulfill the rate condition~\eqref{eqn:KL_v2v}, then no other PMF can satisfy it, hence the corresponding right tree should be discarded.
Otherwise, we can solve the convex optimization problem iteratively until the estimation error ${\mathsf{E}}_\mathcal{C}^{(\ell)} - \mathsf{E}^*$ at iteration $\ell$ reaches below a termination threshold $\varepsilon > 0$, as shown in Algorithm~1.%, where $\mathbb{E}_{\bm{p}^{(\ell)}}(\cdot)$ denotes expectation with respect to $\bm{p}^{(\ell)}$.

%%%%%%
\begin{figure}[!t]
\centering
\includegraphics[width=0.98\linewidth]{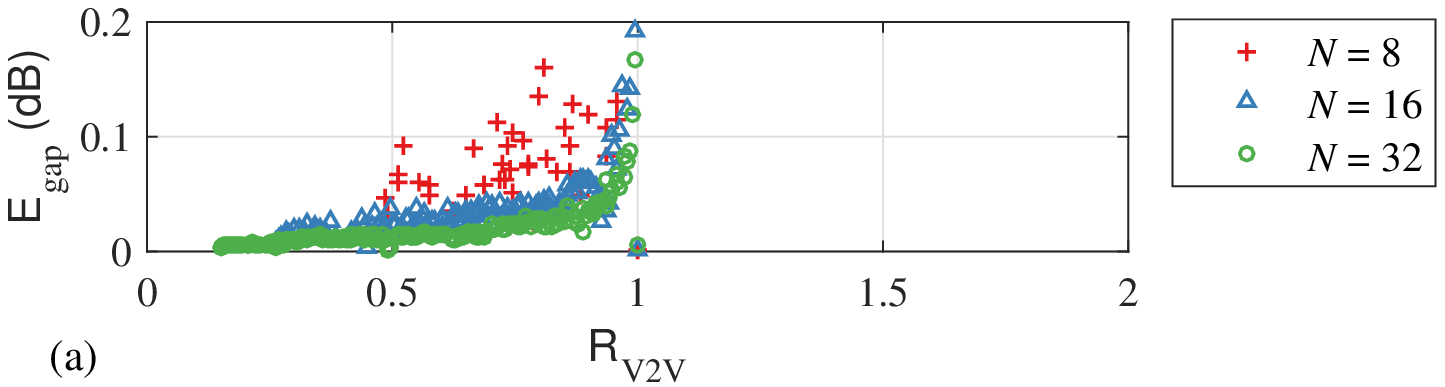}\label{fig:v2v_gap_a} \\[1.2em]
\includegraphics[width=0.98\linewidth]{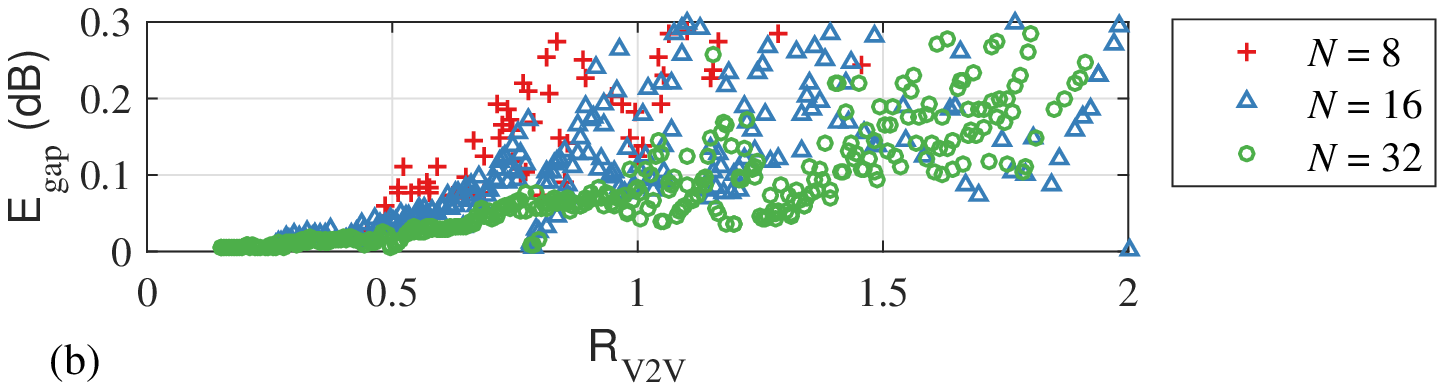}\label{fig:v2v_gap_b} 
\caption{Energy gap of V2V codes with $|\mathcal{C}|\leq 32$ for the (a) 2-ASK, and (b) 4-ASK alphabets.}%
\label{fig:v2v_gap}
\end{figure}
%%%%%%

%%%
\begin{thm}\label{thm_alg_opt}
Given a right tree $\mathcal{T}(N)$, Algorithm~1 produces PMF $\bm{p}^*$ that is asymptotically optimal in iteration $\ell$ for a resolution rate $\mathsf{R}_\mathcal{C} \geq \mathsf{R}^*$, such that $\mathsf{E}_\mathcal{C}$ approaches $\mathsf{E}^*$.
\end{thm}
%%%
%%%
\begin{proof}%[\unskip\nopunct]
The proof follows a similar structure to that of Proposition~4.5 in \cite[Sec.~4.3.1]{Bocherer_thesis}.
Suppose that Algorithm~1 terminates at the iteration $L$.
Then, $\mathsf{{E}}_\mathcal{C}^{(\ell)}$ converges to $\mathsf{E}_\mathcal{C}^{(L)}$ as $\ell$ increases, since the termination condition on Line~5 is not satisfied at every $\ell < L$, hence $\mathsf{{E}}^{(\ell)}_\mathcal{C} \leq \mathsf{{E}}_\mathcal{C}^{(\ell-1)} - \varepsilon$ for $\ell < L$, indicating that $\mathsf{{E}}_\mathcal{C}^{(\ell)}$ is a monotonically decreasing function of $\ell$, while $\mathsf{{E}}_\mathcal{C}^{(\ell)}$ is lower-bounded by $\mathsf{E}^*$.
Furthermore, the converged energy $\mathsf{E}_\mathcal{C}^{(L)}$ is indeed the minimum energy $\mathsf{E}^*$ within a small error.
To see this, let $\Delta^{(\ell)} := \min_{\bm{p}} [\mathbb{E}(||\bm{X}||^2) - \mathsf{E}^{(\ell-1)}_\mathcal{C} \cdot \mathbb{E}(l(\bm{X}))] $ be the optimal value of the objective function at the iteration $\ell$ such that $\mathbb{E}(||\bm{X}||^2) - \mathsf{E}^{(\ell-1)}_\mathcal{C} \cdot \mathbb{E}(l(\bm{X})) \geq \Delta^{(\ell)}$ for any PMF $\bm{p}$, where the equality holds if and only if $\bm{p} = \bm{p}^{(\ell)}$.
Then, $\Delta^{(\ell)} \leq 0$ since $\mathsf{E}_\mathcal{C}^{(\ell-1)} \geq \mathsf{E}^*$ for any PMF $\bm{p}$, where the equality holds if and only if $\mathsf{E}_\mathcal{C}^{(\ell-1)} = \mathsf{E}^*$ by \eqref{eqn:opt_v2v_nec_suf}.
Also, since $\mathbb{E}_{\bm{p}^{(\ell)}}(||\bm{X}||^2) - \mathsf{E}_\mathcal{C}^{(\ell-1)} \cdot \mathbb{E}_{\bm{p}^{(\ell)}}(l(\bm{X})) = \Delta^{(\ell)} \Leftrightarrow \mathsf{E}_\mathcal{C}^{(\ell-1)}  - \mathsf{E}_\mathcal{C}^{(\ell)} = -\Delta^{(\ell)}/\mathbb{E}_{\bm{p}^{(\ell)}}(l(\bm{X}))$, the termination of Algorithm~1 at the iteration $L$ implies $-\Delta^{(L)} / \mathbb{E}_{\bm{p}^{(L)}}(l(\bm{X})) < \varepsilon \Leftrightarrow \Delta^{(L)} > -\varepsilon \cdot \mathbb{E}_{\bm{p}^{(L)}}(l(\bm{X}))  \geq -\varepsilon \cdot \min_{n} l(x_n)$.
Since we have $-\varepsilon \cdot \min_{n} l(x_n) \leq \Delta^{(L)} \leq 0$ while $\Delta^{(L)} = 0$ is a necessary and sufficient condition for $\mathsf{E}_\mathcal{C}^{(L-1)} = \mathsf{E}^*$, Algorithm~1 can closely approach the minimum energy $\mathsf{E}^*$ by choosing a small $\varepsilon$.
\end{proof}
%%%
\noindent In our V2V code construction, Algorithm~1 is terminated mostly in 5 iterations with $\varepsilon = 10^{-10}$.

Once $\bm{p}^*$ is identified for every optimal right tree and for every desired $\mathsf{R}^* = k \Delta_{\mathsf{R}}$, an optimal dyadic estimate of $\bm{p}^*$ can be obtained by $\bm{p} = \text{GHC}(\bm{p}^*)$, which creates a left tree.
Then, among all pairs of such constructed left and right trees, a pair can be chosen for each resolution rate that achieves the minimum $\mathsf{E}_\mathcal{C}$.

Figure~\ref{fig:v2v_gap} shows all V2V codes enumerated with the rate granularity $\Delta_{\text{R}} = 0.005$ and the rate tolerance $\delta = 0.0025$ for the 2-ASK and 4-ASK alphabets with $N = 8,16,32$.
For the 2-ASK alphabet, only $N = 16$ is required to construct V2V codes that have energy gaps smaller than 0.05~dB across a wide range of the resolution rates.
For the 4-ASK alphabet, the energy gap is below 0.2~dB across a wide range of the resolution rates and the realized resolution rates have coarser granularity than those of the 2-ASK alphabet with the same cardinality of the codebook.
%%%%%%%
\begin{figure}[!t]
\centering
\includegraphics[width=0.99\linewidth]{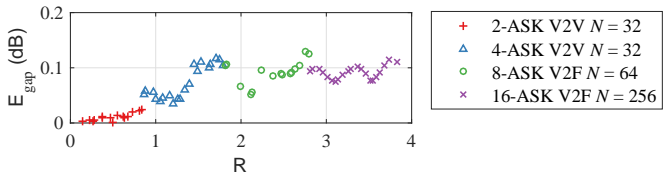}
\caption{Selected prefix-free codes for the 2-, 4-, 8-, and 16-ASK alphabets.} %The entropy rate $\mathsf{R}$ increments in a step size $\leq 0.16$ for rate-adaptable PCS.}
\label{fig:selected_codes}
\end{figure}
%%%%%%%
Some codes selected from Figs.~\ref{fig:v2f_gap} and \ref{fig:v2v_gap} are shown in Fig.~\ref{fig:selected_codes}, whose resolution rates range from 0.15 to 3.83 in a step size $\leq 0.16$.
This shows that we can approach the theoretic minimum energy per code symbol to within 0.13 dB across a wide range of resolution rates for up to 1024-QAM, using codebooks of a cardinality not larger than 256.

 %,in conjunction with a framing method described in the next section.

\section{Framing for Fixed-Rate Transmission}
\label{sec:framing}

In a communication system that needs to carry $k$ information bits using a fixed length-$n$ \emph{frame} of code symbols, a fundamental problem of PCDM is that the codeword length varies depending on the information bits.
Here, and throughout the paper, the frame refers to a container of code symbols, but not of information bits.
Let us define the resolution rate for each of the codewords and for each of the code symbols, as we define the resolution rate for a prefix-free code.
By abuse of terminology, if we refer to them as the \emph{codeword resolution rate} and the \emph{symbol resolution rate}, respectively, these resolution rates are realized at the output of the encoder in a random fashion, and their mean value approaches to $\mathsf{R}_\mathcal{C}$ asymptotically in the number of encoding iterations.
As a consequence, an \emph{overflow} can occur if incoming bits are frequently mapped to codewords whose resolution rate is smaller than $\mathsf{R}_\mathcal{C}$ such that all the code symbols for the $k$ information bits cannot be accommodated in the length-$n$ frame~\cite{Jelinek68buffer}.
On the other hand, if the encoder frequently produces codewords of a resolution rate greater than $\mathsf{R}_\mathcal{C}$, an \emph{underflow} can occur such that part of the frame remains unfilled after completing the encoding.
%The overflow critically hampers application of PCDM to data transmission, and the underflow causes waste of channel use.}
\emph{Framing} in this paper refers to a method that enables a fixed-length frame to always contain a fixed number of information bits, thereby solving the variable length problem of PCDM.
In \cite[Sec.~4.8]{Amjad_thesis}, framing is performed by casting overflow symbols into errors; the probability of error decreases as the frame length grows, but zero-error prefix-free decoding is not possible within a finite-length frame.
%To the best of our knowledge, there is no framing method known to date that allows unique error-free decoding of prefix-free codes.
In~\cite{Yoshida18low}, framing of run-length codes is achieved by making zeros and ones appear exactly the same number of times at the input of the run-length encoder, but the method to construct run-length codes for target resolution rates is not shown.

%This section is devoted to framing of prefix-free codes that facilitates application of PCDM to data transmission%, using some selected codes shown in Fig.~\ref{fig:selected_codes} that are constructed as described in Section~\ref{sec:codes}.
%and to analyze the penalty caused by such framing.

\subsection{Algorithm for Framing}

To achieve a fixed resolution rate $R_{\text{Frame}} < \log_2 M$ in each fixed-length frame for the $M$-ASK alphabet, we use two different codes: a prefix-free code $\mathcal{C}_1 :\bmcal{B}_1 \mapsto \bmcal{X}_1$ with unequal codeword resolution rates, whose mean value $\mathsf{R}_{\mathcal{C}_1}$ is close to $R_{\text{Frame}}$, and a trivial code $\mathcal{C}_2 : \bmcal{B}_2 \mapsto \bmcal{X}_2$ with an equal codeword resolution rate $\mathsf{R}_{\mathcal{C}_2} = \log_2{M} > \mathsf{R}_{\mathcal{C}_1}$ across all the codewords (i.e., $\mathcal{C}_2$ is a typical F2F mapper for uniform $M$-ASK).
The idea is that we begin encoding with $\mathcal{C}_1$ and then \emph{switch} to $\mathcal{C}_2$ at some point of the successive encoding process if an overflow is predicted.
If we keep counting the numbers of input bits and output symbols encoded by $\mathcal{C}_1$ in the first part of encoding, we can also calculate the number of symbols required to encode all the remaining input bits if we use $\mathcal{C}_2$ instead of $\mathcal{C}_1$ from that point onward, enabling prediction of an overflow.
In what follows, we will show that this prediction can be made in a way that unique decoding is possible, and that the penalty due to this switching is small. 

Let $k$ and $n$ be, respectively, a fixed number of input bits and a fixed number of output symbols of a PCDM encoder in each frame; i.e., the information bits and the code symbols belong to $\mathcal{B}^k$ and $\mathcal{X}^n$, respectively.
Then, framing enables PCDM to achieve a fixed ${R}_{\text{Frame}}$, where
\begin{align}
\mathsf{R}_{\mathcal{C}_1} < {R}_{\text{Frame}} := \frac{k}{n} < \mathsf{R}_{\mathcal{C}_2}.	\label{eqn:Rframe1}
\end{align}
This shows that, for the chosen $\mathsf{R}_{\mathcal{C}_1}$ and $\mathsf{R}_{\mathcal{C}_2}$, an additional rate adaptability is also offered by framing, at the expense a larger energy gap.
Let $\bm{b}^{(\ell)}$ and $\bm{x}^{(\ell)}$ denote the information word and the codeword chosen at the encoding iteration~$\ell$ by $\mathcal{C}_1$, respectively.
And assume that the use of $\mathcal{C}_1$ until iteration~$\ell-1$ was assured not to cause an overflow, as long as we switch to $\mathcal{C}_2$ from the next iteration~$\ell$ onwards.
Hence we have used $\mathcal{C}_1$ until iteration~$\ell-1$.
Then, at the next iteration~$\ell$, in order to foresee if $\mathcal{C}_1$ does not still cause an overflow, we need to ensure that
\begin{align}
\gamma_{\text{Ava}}^{(\ell)} \geq \gamma_{\text{Req.}}^{(\ell)},    \label{eqn:cond_no_overflow}
\end{align}
where $\gamma_{\text{Ava}}^{(\ell)} := n - \sum_{i=1}^{\ell} l(\bm{x}^{(i)})$ is the number of available slots in the frame {at the beginning of iteration~$\ell+1$} and $\gamma_{\text{Req.}}^{(\ell)} := \left\lceil (k - \sum_{i=1}^{\ell} l(\bm{b}^{(i)}) )/ {\mathsf{R}_{\mathcal{C}_2}} \right\rceil$ is the number of required symbol slots in case we switch to $\mathcal{C}_2$ at iteration~$\ell+1$, respectively.
Codebook $\mathcal{C}_1$ is used for encoding at iteration $\ell$ if condition \eqref{eqn:cond_no_overflow} is fulfilled, otherwise $\mathcal{C}_2$ from iteration $\ell$ onwards.
%If Eq.~\eqref{eqn:cond_no_overflow} is satisfied, then we can use $\mathcal{C}_1$ at iteration~$\ell$ without concerning overflow; otherwise, we should avoid overflow by switching to $\mathcal{C}_2$ at iteration $\ell$, where avoidance of overflow by this switching has been ensured by assumption from the previous iteration $\ell-1$.
Notice that \eqref{eqn:cond_no_overflow} can be evaluated only after seeing the incoming bits at iteration~$\ell$ to obtain $l(\bm{b}^{(\ell)})$ and $l(\bm{x}^{(\ell)})$.
This makes unique decoding impossible, since, assuming that unique decoding was successfully performed until iteration~$\ell-1$ such that $l(\bm{b}^{(i)})$ and $l(\bm{x}^{(i)})$ are known for all $i=1,\ldots,\ell-1$, the decoder cannot identify which codebook was used at iteration~$\ell$ without knowing $l(\bm{b}^{(\ell)})$ and $l(\bm{x}^{(\ell)})$.
This suggests that, for unique decoding, the codebook used at iteration~$\ell$ must be identified without relying on $l(\bm{b}^{(\ell)})$ and $l(\bm{x}^{(\ell)})$, which can be realized by a bounding technique.
Namely, in a \emph{pessimistic} assumption that the shortest information word is mapped to the longest codeword of $\mathcal{C}_1$ at iteration $\ell$, we can find the lower bound of the available slots and the upper bound of the required slots as
\begin{align*}
\underline{\gamma}_{\text{Ava}}^{(\ell)}   &:= n - \sum_{i=1}^{\ell-1} l(\bm{x}^{(i)}) - l_{\max}(\bmcal{X}),			\eqnum \label{eqn:gamma_ava2}
\end{align*}
and
\begin{align*}
\overline{\gamma}_{\text{Req}}^{(\ell)}    &:= \left\lceil \frac { k - \sum_{i=1}^{\ell-1} l(\bm{b}^{(i)}) - l_{\min}(\bmcal{B}) } {\mathsf{R}_{\mathcal{C}_2}} \right\rceil,  \eqnum \label{eqn:gamma_req2}
\end{align*}
where $l_{\max}(\bmcal{X}) := \max_{\bm{x}\in \bmcal{X}} l(\bm{x})$ and $l_{\min}(\bmcal{B}) := \max_{\bm{b}\in \bmcal{B}} l(\bm{b})$.
Now, without knowledge of $l(\bm{b}^{(\ell)})$ and $l(\bm{x}^{(\ell)})$, the condition~\eqref{eqn:cond_no_overflow} can be conservatively examined by
\begin{align}
\underline{\gamma}_{\text{Ava}}^{(\ell)} \geq \overline{\gamma}_{\text{Req}}^{(\ell)},    \label{eqn:cond_no_overflow2}
\end{align}
since $\gamma_{\text{Ava}}^{(\ell)} \geq \underline{\gamma}_{\text{Ava}}^{(\ell)}$ and $\overline{\gamma}_{\text{Req}}^{(\ell)} \geq \gamma_{\text{Req}}^{(\ell)}$.
%%%
\begin{thm}\label{thm_cond_overflow}
\label{thm:no_overflow}
In the successive PCDM encoding process, an overflow can be avoided by switching the code from $\mathcal{C}_1$ to $\mathcal{C}_2$ at the earliest iteration $\ell$ that does not fulfill \eqref{eqn:cond_no_overflow2}.
\end{thm}
%%%
\begin{proof}%[\unskip\nopunct]
If \eqref{eqn:cond_no_overflow2} is fulfilled at iteration $\ell = 1$, we can use $\mathcal{C}_1$ at $\ell = 1$ without an overflow, otherwise we can use $\mathcal{C}_2$ to encode all the information bits without an overflow, since $\lceil k / \mathsf{R}_{\mathcal{C}_2} \rceil \leq n$ by \eqref{eqn:Rframe1}.
Assume that $\mathcal{C}_1$ was used at iteration $\ell-1$, fulfilling \eqref{eqn:cond_no_overflow2}, hence ensured that an overflow will not occur if we switch to $\mathcal{C}_2$ at iteration $\ell$.
If \eqref{eqn:cond_no_overflow2} is fulfilled also at iteration $\ell$, we keep using $\mathcal{C}_1$, since encoding with $\mathcal{C}_2$ from iteration $\ell+1$ suffices to avoid an overflow, otherwise, we switch to $\mathcal{C}_2$ at iteration $\ell$ without an overflow by assumption. This completes proof by mathematical induction.
\end{proof}
%%%

%%%
\begin{thm}\label{thm_unique_decoding}
The codewords framed by using Theorem~\ref{thm:no_overflow} can be uniquely decoded.
\end{thm}
%%%
\begin{proof}%[\unskip\nopunct]
At iteration $\ell = 1$, it is trivial to see that a decoder can identify a decoding code using \eqref{eqn:cond_no_overflow2}, which allows for unique decoding and gives knowledge of $l(\bm{b}^{(1)})$ and $l(\bm{x}^{(1)})$.
Assume that $l(\bm{b}^{(i)})$ and $l(\bm{x}^{(i)})$ for all $i=1,\ldots,\ell-1$ are known at the beginning of iteration $\ell$. Then, by assessing \eqref{eqn:cond_no_overflow2} at iteration $\ell$, we identify the decoding code, enabling unique decoding at iteration $\ell$. This provides $l(\bm{b}^{(\ell)})$ and $l(\bm{x}^{(\ell)})$ and completes proof by mathematical induction.
\end{proof}

%%%%%%%%%%%%%%%%%%
\begin{algorithm}[!t]
\label{alg:framing}
\caption{Framing of a prefix-free code}
\begin{algorithmic}[1]
\renewcommand{\algorithmicrequire}{\textbf{Input:}}
\renewcommand{\algorithmicensure}{\textbf{Output:}}
\REQUIRE $\mathcal{C}_1,\mathcal{C}_2$, an information bit sequence $\bm{b} \in \mathcal{B}^k$
\ENSURE  A code symbol frame $\bm{x} \in \mathcal{X}_{M\text{-ASK}}^n$
\\ \textit{Initialization}: $\ell\leftarrow0$.
\REPEAT
	\STATE $\ell \leftarrow \ell + 1$
	\STATE Find $\bm{b}^{(\ell)}\in\bmcal{B}_1$ and produce the corresponding $\bm{x}^{(\ell)}\in\bmcal{X}_1$ as an output.
%		Encode the current length-$l(\bm{b}^{(\ell)})$ segment of $\bm{b}$ using $\mathcal{C}_1$ and produce the length-$l(\bm{x}^{(\ell)})$ segment of $\bm{x}$.
%		\STATE Encode all the remaining bits in $\bm{b}$ using $\mathcal{C}_2$ to fill $\bm{x}$.
\UNTIL{$\sum_{i=1}^{\ell}l(\bm{b}^{(i)}) <= k$ and $\underline{\gamma}_{\text{Ava}}^{(\ell)} \geq \overline{\gamma}_{\text{Req}}^{(\ell)}$}
\IF{$\sum_{i=1}^{\ell}l(\bm{b}^{(i)}) < k$ }
	\REPEAT
		\STATE $\ell \leftarrow \ell + 1$
		\STATE Find $\bm{b}^{(\ell)}\in\bmcal{B}_2$ and produce the corresponding $\bm{x}^{(\ell)}\in\bmcal{X}_2$ as an output.
	\UNTIL{$\sum_{i=1}^{\ell}l(\bm{b}^{(i)}) <= k$}
\ENDIF
\IF{$\sum_{i=1}^{\ell}l(\bm{x}^{(i)}) < n$}
	\STATE Fill the last $n-\sum_{i=1}^{\ell}l(\bm{x}^{(i)})$ slots in $\bm{x}$ with 1.
\ENDIF
\RETURN $\bm{x}$
\end{algorithmic}
\end{algorithm}
%%%%%%%%%%%%%%%%%%

In cases where an \emph{underflow} occurs, the encoder can simply fill the unoccupied slots in the frame with dummy symbols of the smallest energy, e.g., $\mathcal{X}_1 = 1$ for the $M$-ASK alphabet.
The decoder can discard these dummy symbols after $k$ bits are all decoded.
There may also be incidents at the end of encoding that need minor manipulations as follows.
At the last encoding iteration, it is possible that there is no information word in the dictionary that matches the input bits.
In this case, there must be multiple information words whose prefix matches the input bits, hence an encoder can set a rule for unique encoding; e.g., the first codeword in an lexicographical order can be picked.
Unique decoding is straight-forward. 
A pseudo-code for this framing algorithm is given in Algorithm~2.

\subsection{Analysis of Fixed-Length Penalty by Gaussian Approximation}

The probability that the codeword resolution rate at a certain time instant equals $r_n := l(\bm{b}_n)) / l(\bm{x}_n)$ is the probability that the current output symbol belongs to the codeword $\bm{x}_n \in {\bmcal{X}_1}$, which can be calculated as $q_n := p_n l(\bm{x}_n) / \sum_{i=1}^{N} p_i l(\bm{x}_i)$, where $N = |\bmcal{X}|$.
Without the fixed-length framing constraint, the resolution rate~\eqref{eqn:R} can alternatively be calculated as $\mathsf{R}_{\mathcal{C}_1} = \mathbb{E}_{\bf{q}}(R)$, where $R$ is a random variable taking values on an ordered set $\{r_1, \ldots, r_N \}$, with the PMF $\bm{q} := [q_1,\ldots,q_N]^T$.
For example, with the V2V code in Tab.~\ref{tab:examples}~(c), a codeword resolution rate in $\bm{r} := [r_1,\ldots,r_N]^T \approx [$0.14, 0.43, 0.5, 0.6, 1, 1.67, 3, 6$]^T$ is observed at the encoder output with a PMF $\bm{q} \approx [$0.57, 0.143, 0.122, 0.102, 0.041, 0.015, 0.005, 0.003$]^T$, which yields $\mathsf{R}_{\mathcal{C}_1} \approx 0.361$ on average. 
Suppose now that a fixed-length framing constraint is imposed.
We consider the \emph{ensemble} of frames and assume that all symbols encoded by $\mathcal{C}_1$ are IID.
Also, though the true symbol resolution rate is a discrete random variable, we approximate this by a continuous Gaussian random variable with mean $\mathsf{R}_{\mathcal{C}_1}$ and variance $\mathsf{S}^2 := \mathbb{E}_{\bm{q}} (R^2) - \mathsf{R}_{\mathcal{C}_1}^2$; i.e., $R \sim \mathcal{N} (\mathsf{R}_{\mathcal{C}_1},\mathsf{S}^2)$.
In this case, the number of information bits that are mapped until the $t$-th symbol output follows a Gaussian distribution $\mathcal{N} (t\mathsf{R}_{\mathcal{C}_1},t\mathsf{S}^2)$, if there has been neither an overflow prediction nor an early termination of encoding by $\mathcal{C}_1$ until the $(t-1)$-th symbol.
To take into account a fulfilled overflow prediction and an early termination, we notice that \eqref{eqn:gamma_ava2} and \eqref{eqn:gamma_req2} at the $t$-th output symbol can respectively be transformed into $\underline{\gamma}_{\text{Ava}} (t)   := n - (t-1) - l_{\max}(\bmcal{X})$ and $\overline{\gamma}_{\text{Req}} (t) := \left\lceil { \big( k-\Theta(t-1) - l_{\min}(\bmcal{B}) \big) } / {\mathsf{R}_{\mathcal{C}_2}} \right\rceil$, where $\Theta(t)$ is a random variable representing the \emph{cumulative symbol resolution rate} at the $t$-th output symbol.

%%%%%%
\begin{figure}[!t]
\centering
\hspace*{-2em}\includegraphics[width=2.3in]{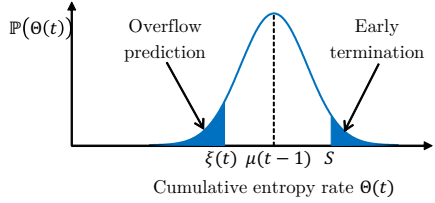}\\[-1em]
%\hspace*{-2em}\includegraphics[width=2.3in]{fig_ga_illustration_cr.pdf}\\[-1em]
\caption{Approximated Gaussian distribution of the cumulative symbol resolution} rate at the $t$-th symbol.
\label{fig:ga_illustration}
\end{figure}
%%%%%%

However, characterization of $\Theta(t)$ becomes infeasible as the overflow and early termination probabilities increase with $t$, hence we approximate $\Theta(t)$ for all $t = 1,\ldots,n$ by a Gaussian random variable with mean $\mu(t)$ and variance $\sigma^2(t)$ such that $\Theta(t) \sim \mathcal{N} (\mu(t), \sigma^2(t))$ whose evolution over $t$ is mathematically tractable, as shown in Fig.~\ref{fig:ga_illustration}.
Then, on condition that $\mathcal{C}_1$ is used at the $t$-th symbol, i.e., on condition that neither an overflow prediction nor an early termination has been made before the $t$-th symbol, the probability of an overflow prediction at the $t$-th symbol is
\begin{align*}
&\mathbb{P} \left( \overline{\gamma}_{\text{Req}} (t) > \underline{\gamma}_{\text{Ava}} (t)   \right) \\
&\approx \mathbb{P} \left(   k -\Theta(t-1) - l_{\min}(\bmcal{B}) - \mathsf{R}_{\mathcal{C}_2} \left( n - t+1 - l_{\max}(\bmcal{X}) \right) > 0 \right) \\
&= \mathbb{P} \left( \Theta(t-1)  < \xi(t)  \right)			\eqnum \label{eqn:overflow_threshold} \\
&= F\left( \xi(t) \,|\,\mu(t-1),\sigma^2(t-1)\right)		\eqnum \label{eqn:ga_overflow_condition}
\end{align*}
where, in \eqref{eqn:overflow_threshold}, an overflow threshold is defined as $\xi(t) := \mathsf{R}_{\mathcal{C}_2} \left( t - n - 1 + l_{\max}(\bmcal{X}) \right) + k - l_{\min}(\bmcal{B})$, and $F\left( \,\cdot \,|\,\mu,\sigma^2\right)$ in \eqref{eqn:ga_overflow_condition} denotes Gaussian cumulative distribution function (CDF) with mean $\mu$ and variance $\sigma^2$.
Equation~\eqref{eqn:overflow_threshold} gives insight into the behavior of the proposed framing method, since it indicates that an overflow is predicted if the symbol resolution rate accumulated until the $(t-1)$-th symbol is smaller than the threshold $\xi(t)$ that increases linearly with $t$.
Here, $\mu(t)$ and $\xi(t)$ are both linear in $t$ with slopes $\mathsf{R}_{\mathcal{C}_1}$ and $\mathsf{R}_{\mathcal{C}_2}$, respectively, and since $\mathsf{R}_{\mathcal{C}_1} < \mathsf{R}_{\mathcal{C}_2}$ by the framing rule, the probability of an overflow prediction in \eqref{eqn:overflow_threshold} gradually increases with $t$.

%%%%%%
\begin{figure}[!t]
\centering
\includegraphics[width=3.4in]{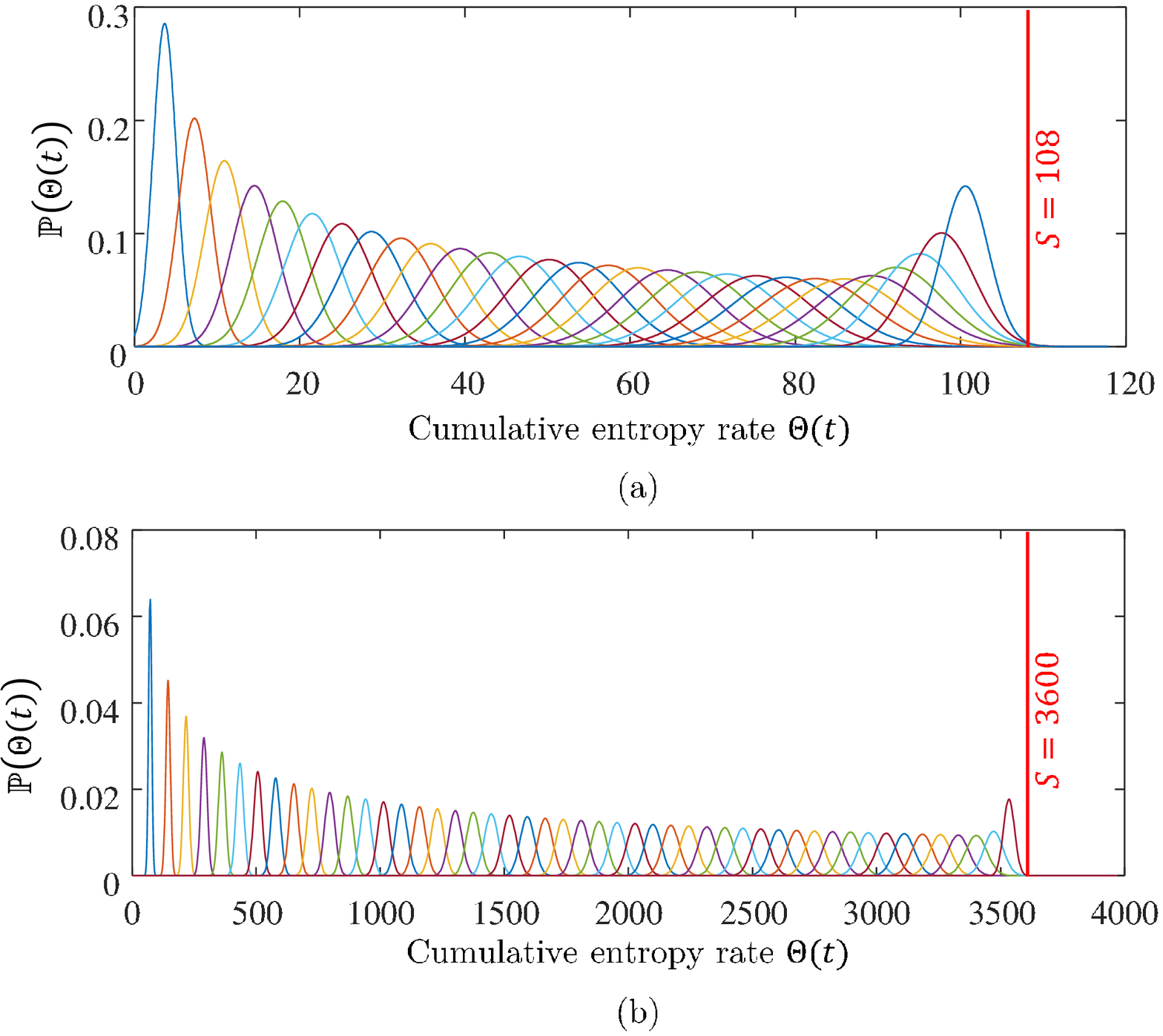}
\caption{Evolution of the distribution of the cumulative symbol resolution rate: (a) $k = 108, n = 300, t = 10i$, and (b) $k = 3600, n = 10000, t = 200i$, for positive integers $i$.}
\label{fig:ga_evolution}
\end{figure}
%%%%%%

%%%%%%
\begin{figure}[!t]
\centering
\subfloat[][]{\!\!\!\!\!\includegraphics[width=1.12\linewidth]{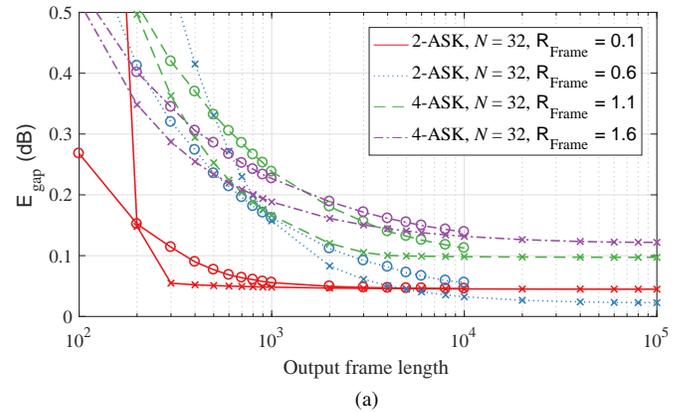}\label{fig:fixed_length_penalty1}}\\[-0.1em]
\subfloat[][]{\!\!\!\!\!\includegraphics[width=1.12\linewidth]{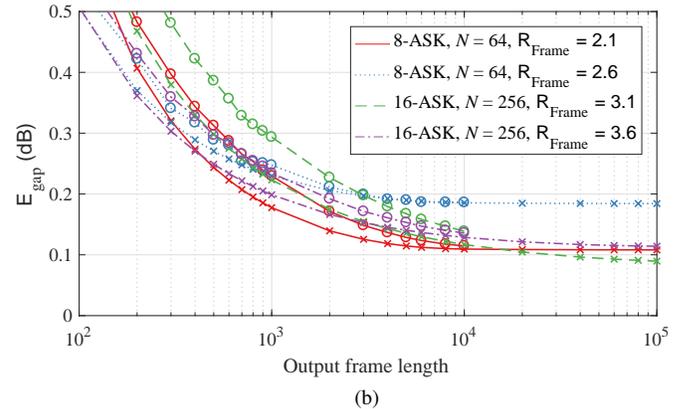}\label{fig:fixed_length_penalty2}}
%\subfloat[][]{\includegraphics[width=0.5\linewidth]{fig_fixed_length_penalty1_cr.pdf}\label{fig:fixed_length_penalty1}}
%\subfloat[][]{\includegraphics[width=0.5\linewidth]{fig_fixed_length_penalty2_cr.pdf}\label{fig:fixed_length_penalty2}}\\[-0.5em]
\caption{Energy gap after framing of (a) V2V codes for the 2- and 4-ASK alphabets, and (b) V2F codes for the 8- and 16-ASK alphabets. The cross and circle markers show the results of GA analysis and MC simulations, respectively.}% 
\label{fig:fixed_length_penalty}
\end{figure}
%%%%%%

In order to take into account the history of previous overflow predictions and early terminations, let $\Phi_{\rm{Swi}}(t)$ and $\Phi_{\rm{End}}(t)$ denote the cumulative overflow prediction probability and cumulative early termination probability at symbol~$t$.
Then, the \emph{unconditional} probability that an overflow is predicted at symbol~$t$ is 
\begin{align*}
\phi_{\rm{Swi}}(t) := \big(1 - \Phi_{\rm{Swi}}(t-1) - \Phi_{\rm{End}}(t-1)\big) \\
				     \times F\left(\xi(t) \,|\, \mu(t-1),\sigma^2(t-1)\right). 		\eqnum \label{eqn:ga_overflow_pr}
\end{align*}
And the probability of an early termination at symbol~$t$ is 
\begin{align*}
\phi_{\rm{End}}(t) := \big(1 - \Phi_{\rm{Swi}}(t-1) - \Phi_{\rm{End}}(t-1) \big) \\
				    \times \left( 1- F(S \,|\, \mu(t-1),\sigma^2(t-1)) \right).			\eqnum \label{eqn:ga_early_termination_pr}
\end{align*}
It is straightforward to see that their cumulative probabilities are obtained by $\Phi_{\rm{Swi}}(t) = \sum_{i=1}^{t} \phi_{\rm{Swi}}(t)$ and $\Phi_{\rm{End}}(t) = \sum_{i=1}^{t} \phi_{\rm{End}}(t)$.
From the initial conditions $\Phi_{\rm{Swi}}(0) = 0$ and $\Phi_{\rm{End}}(0) = 0$, the cumulative probabilities can be evaluated in an iterative manner from symbol $1$.
When evaluating \eqref{eqn:ga_overflow_pr} and \eqref{eqn:ga_early_termination_pr}, we have the initial conditions $\mu(0) = 0$ and $\sigma^2(0) = 0$, and $\mu(t)$ and $\sigma^2(t)$ for $t > 0$ can be obtained using the update rule (see Fig.~\ref{fig:ga_illustration}) $\mu(t) = \int_{\xi(t-1)}^{k} \theta \frac {f\left(\theta \,|\, \mu(t-1),\sigma^2(t-1)\right)}{\alpha} d\theta + \mathsf{R}_{\mathcal{C}_1}$, and $\sigma^2(t) = \int_{\xi(t-1)}^{k} (\theta-\mu(t))^2 \frac{f\left(\theta \,|\, \mu(t-1),\sigma^2(t-1)\right)}{\alpha} d\theta + \mathsf{S}^2$, where $f\left( \,\cdot \,|\,\mu,\sigma^2\right)$ denotes Gaussian probability density function (PDF) with mean $\mu$ and variance $\sigma^2$, and $\alpha$ is a normalization factor defined by $\alpha := \int_{\xi(t-1)}^{k} f\left(\theta \,|\, \mu(t-1),\sigma^2(t-1)\right) d\theta$.
$\mu(t)$ and $\sigma^2(t)$ can be efficiently evaluated by numerical integration methods \cite{davis2007methods}.

%%%%%%%
\begin{figure*}[!t]
\centering
\captionsetup{skip=-5pt}
\includegraphics[width=0.81\linewidth]{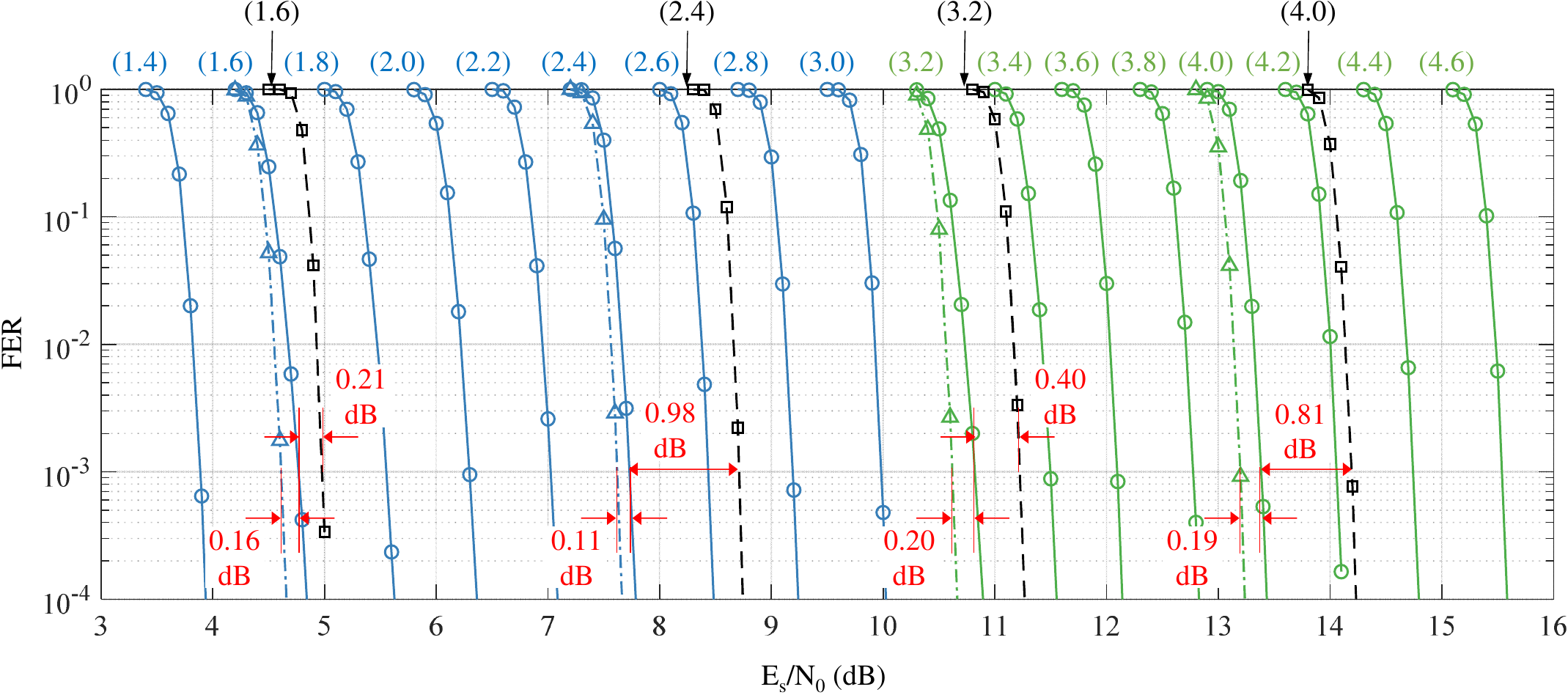}\\[-0.1em]
\caption{FER of uniform QAM (square markers) and PS-QAM realized by CCDM (triangle markers), and PCDM (circle markers). The numbers in the parentheses show the IRs.}
\label{fig:snr_fer}
\end{figure*}
%%%%%%%

Figure~\ref{fig:ga_evolution} shows the evolution of $\Theta(t)$, when we use the aforementioned Gaussian approximation (GA) for the V2V code in Tab.~\ref{tab:examples}~(c).
It realizes $\mathsf{R}_{\text{Frame}} = 0.36$ in each frame, using the probabilistic symbol resolution rate $R$ with mean $\mathsf{R}_{\mathcal{C}_1} \approx 0.361$ and variance $\mathsf{S}^2 = 0.195$.
Each small Gaussian PDF represents a PDF of the cumulative symbol resolution rate at the output symbol $t$, where $t$ increments from 10 in a step size of 10 in~(a), and from 200 in a step size of 200 in~(b).
The variance of the Gaussian PDF grows with $t$ at first, which is expected with little effect from the fixed-length constraint, then begins to decrease near the end of the frame as the probabilities of an overflow prediction and an early termination emerge.
Comparison of Figs.~\ref{fig:ga_evolution}~(a) and (b) shows that a large frame length reduces the impact of framing on the increase of the energy gap, as will be quantitatively shown below.

The average symbol energy under framing can be estimated by GA as
\begin{align*}
%\mathsf{E}_{\text{Frame}} = \frac{1}{T}\sum_{t=1}^{T} &\Big[  \left( 1- \Phi_{\rm{Swi}}(t) - \Phi_{\rm{End}}(t-1) \right) \mathsf{E}_1 \\
%& + \Phi_{\rm{Swi}}(t) \mathsf{E}_2 + \Phi_{\rm{End}}(t-1) \Big],  \eqnum \label{eqn:E_frame}
&\mathsf{E}_{\text{Frame}} = \frac{1}{n}\sum_{t=1}^{n} \Big[  \left( 1- \Phi_{\rm{Swi}}(t) - \Phi_{\rm{End}}(t-1) \right)  
\mathsf{E}_{\mathcal{C}_1}  \\
&\qquad\qquad\qquad + \Phi_{\rm{Swi}}(t) \mathsf{E}_{\mathcal{C}_2} + \Phi_{\rm{End}}(t-1) \Big].  \eqnum \label{eqn:E_frame}
\end{align*}
The energy gap of the prefix-free codes is depicted in Fig.~\ref{fig:fixed_length_penalty}, where the codes are selected to support various $\mathsf{R}_{\text{Frame}}$, ranging from 0.1 to 3.6 in a step size of 0.5.
The energy gap is estimated by GA, and also by Monte Carlo (MC) simulations, averaged over 10000 frames created from random equiprobable information bits.
GA provides very close results to MC simulations as the frame length increases, in a significantly shorter time.
Since a smaller frame length allows a greater flexibility in implementation of the PCDM and parallel processing of multiple frames, Fig.~\ref{fig:fixed_length_penalty} shows there exists a trade-off between the frame length and the energy gap. 
In the limit of the frame length, prefix-free codes approach the ideal energy efficiency of the MB PMFs to within 0.2~dB, even with framing.

\section{Rate-Adaptable PCDM in AWGN Channel}

Figure~\ref{fig:snr_fer} shows the performance of rate-adaptable PCS in the AWGN channel, realized by PCDM in the PAS architecture.
PCDM is implemented with frame lengths of $30720/\log_2 M$ QAM symbols, and a tail-biting (TB) spatially-coupled low-density parity-check (SC-LDPC) code of length 30720 bits and rate $\mathsf{r}_c = 4/5$ is used for FEC, constructed based on \cite{cho2015construction}.
In the simulation, one PCDM frame is made to be exactly one FEC frame, hence the frame error rate (FER) is the same for PCDM decoding and FEC decoding.
The TB-SC-LDPC code is 0.9~dB away from the binary phase-shift keying (BPSK) AWGN capacity at the FER of $10^{-3}$, when decoded by the normalized min-sum algorithm with $\leq$ 30 ierations.
V2V prefix-codes are used for PCDM, which realize the DM rates $\mathsf{R}_{\text{Frame}} = 0.1, 0.2, \ldots, 0.9$ for the $2$-ASK alphabet and $\mathsf{R}_{\text{Frame}} = 1.2, 1.3, \ldots, 1.9$ for the $4$-ASK alphabet.
Each V2V code has only 32 codewords in the codebook.
The IR of probabilistically shaped (PS)-$M$-QAM is given by $\mathsf{R}_{\text{PS-}M\text{-QAM}} = 2 \left[ 1+\mathsf{R}_{\text{Frame}}-(1-\mathsf{r}_c)\frac{\log_2 M}{2} \right]$, where $\mathsf{r}_c$ denotes the FEC code rate\cite[eq.~(30)]{Bocherer15}.
Note that $\mathsf{R}_{\text{PS-}M\text{-QAM}}$ is quantified in bits/QAM symbol, whereas all the resolution rates in this paper are in bits/ASK symbol.
Using the selected V2V codes and $\mathsf{r}_c = 4/5$, the IRs of $\mathsf{R}_{\text{PS-}16\text{-QAM}} = 1.4, 1.6, \ldots, 3.0$ and $\mathsf{R}_{\text{PS-}64\text{-QAM}} = 3.2, 3.4, \ldots, 4.6$ are obtained, as shown by the parentheses in Fig.~\ref{fig:snr_fer}.
The FER is estimated as a function of $\mathsf{E_s/N_0}$, with $\mathsf{E_s}$ being energy per channel use and $N_0$ being noise variance per two dimensions.
For comparison, also shwon is the performance of the uniform 4-, 8-, 16-, and 32-QAM.
The constellations of the uniform 4- and 16-QAM are the square QAM constellations with Gray labeling, and those of the uniform 8- and 32-QAM are `Circular-8QAM' of \cite[Ch.~3.1.1]{Muller16advanced} and $\bm{X}_{32,\text{cro}}$ of \cite[Ch.~3.2.4]{Muhammad10Coding}, respectively.
It can be seen that PCDM approaches the performance of CCDM to within 0.2~dB at an FER of $10^{-3}$.
Note that, with the chosen frame lengths, CCDM is very close to an ideal DM \cite{Schulte16}.
As a result, PCDM achieves shaping gains up to 0.98~dB over the uniform QAM for the same IR.

\section{Conclusion}

We created a wide range of resolution rates with a very fine granularity using tiny prefix-free codes.
In particular, we constructed optimal V2F codes of cardinality $\leq$ 4096 for the 2-, 4-, 8-, and 16-ASK alphabets, optimal in the sense that no other code of the same cardinality can achieve a smaller average symbol energy for the same resolution rate.
We also enumerated all of the optimal F2V codes and many near-optimal V2V codes of cardinality $\leq$ 32 with the 2- and 4-ASK alphabets.
Selected codes approach the theoretic lowest energy to within 0.13 dB across the resolution rate from 0.1 to 3.9~bits/symbol in a step size less than 0.16~bits/symbol, showing that PCDM is suitable for rate-adaptable PCS.
Its implementation cost is a small look-up table for each resolution rate.
We also proposed a framing method to implement fixed-rate transmission with variable-rate prefix-free coding, and proved that the proposed framing enables unique prefix-free decoding.
We gave insights into the evolution of encoding process under framing.
Using GA analysis and MC simulations, we showed that the penalty caused by framing is negligible when the frame length is large.
FEC simulations in the AWGN channel demonstrate that the PCDM-based PS-QAM under framing achieves an SNR gain up to 0.98~dB compared to uniform QAM.

\end{document}